\documentclass[12pt,reqno]{amsart}
\usepackage{fullpage}

\newtheorem{theorem}{Theorem}[section]
\newtheorem{lemma}[theorem]{Lemma}

\newtheorem{conj}[theorem]{Conjecture}

\usepackage{graphicx}
\usepackage{color}
\usepackage[dvipsnames]{xcolor}
\usepackage{subfigure}
\usepackage{amssymb}
\usepackage{amsmath,mathrsfs}
\usepackage{colonequals}
\usepackage[backref=page]{hyperref}

\theoremstyle{definition}
\newtheorem{definition}[theorem]{Definition}

\newtheorem{remark}[theorem]{Remark}

\renewcommand{\subset}{\subseteq}

\renewcommand{\epsilon}{\varepsilon}

\newcommand{\abs}[1]{\left|#1\right|}                   
\newcommand{\vnorm}[1]{\left\|#1\right\|}    
\newcommand{\vnormf}[1]{\|#1\|}                         

\newcommand{\N}{\mathbb{N}}
\newcommand{\E}{\mathbb{E}}

\newcommand{\R}{\mathbb{R}}
\newcommand{\C}{\mathbb{C}}

\renewcommand{\d}{\mathrm{d}}

\newcommand{\embolden}[1]{\textbf {#1}}

\newcommand{\sdimn}{n}
\newcommand{\adimn}{n}

\begin{document}

\title{Sphere Valued Noise Stability and Quantum MAX-CUT Hardness}

\author{Steven Heilman}
\address{Department of Mathematics, University of Southern California, Los Angeles, CA 90089-2532}
\email{stevenmheilman@gmail.com}
\date{\today}
\keywords{noise stability, Borell inequality, MAX-CUT, Quantum MAX-CUT}

\begin{abstract}
We prove a vector-valued inequality for the Gaussian noise stability (i.e. we prove a vector-valued Borell inequality) for Euclidean functions taking values in the two-dimensional sphere, for all correlation parameters at most $1/10$ in absolute value.  This inequality was conjectured (for all correlation parameters at most $1$ in absolute value) by Hwang, Neeman, Parekh, Thompson and Wright.  Such an inequality is needed to prove sharp computational hardness of the product state Quantum MAX-CUT problem, assuming the Unique Games Conjecture.  In fact, assuming the Unique Games Conjecture, we show that the product state of Quantum MAX-CUT is NP-hard to approximate within a multiplicative factor of $.9859$.  In contrast, a polynomial time algorithm is known with approximation factor $.956\ldots$.
\end{abstract}
\maketitle

%
%
%
%
\section{Introduction}\label{secintro}

The noise stability of a measurable Euclidean set $A$ with correlation $\rho$ is the probability that $(X,Y)\in A\times A$, where $X,Y$ are standard Gaussian random vectors with correlation $\rho\in(-1,1)$.  Borell's inequality asserts that half spaces have the largest noise stability among all Euclidean sets of fixed Gaussian measure.  Borell's inequality \cite{borell85} generalizes the Gaussian isoperimetric inequality, since letting $\rho\to1^{-}$ in Borell's inequality recovers the Gaussian isoperimetric inequality \cite{ledoux94}.  The Gaussian isoperimetric inequality says that a half space has the smallest Gaussian surface area among all Euclidean sets of fixed Gaussian volume.

Besides its intrinsic interest, Borell's inequality has been applied to social choice theory \cite{mossel10}, the Unique Games Conjecture \cite{khot07,mossel10,khot15}, to semidefinite programming algorithms such as MAX-CUT \cite{khot07,isaksson11}, to learning theory \cite{feldman12}, etc.  For some surveys on this and related topics, see  \cite{odonnell14b,khot10b,heilman20b}.

A Corollary of Borell's inequality is the Majority is Stablest Theorem \cite{mossel10}.  Moreover, Borell's inequality was the main technical ingredient used to prove sharp computational hardness of the MAX-CUT problem, assuming the Unique Games Conjecture \cite{khot07}.

The MAX-CUT problem asks for the partition of the vertices of an undirected finite graph into two disjoint sets that maximizes the number of edges going between the two sets.  The MAX-CUT problem is a well-studied NP-complete constraint satisfaction problem with well-understood hardness of approximation \cite{goemans95,khot07}.  The Unique Games Conjecture implies that the Goemans-Williamson semidefinite program is the best quality polynomial time algorithm for approximately solving MAX-CUT.

The quantum analogue of a constraint satisfaction problem is a local Hamiltonian problem, which is QMA-complete \cite{gharibian12}.  The Quantum MAX-CUT problem is a special case of the local Hamiltonian problem, that in some sense generalizes the usual MAX-CUT problem.  As with MAX-CUT, it is natural to ask for approximation algorithms for Quantum MAX-CUT, and to try to prove sharp computational hardness of those algorithms \cite{hwang21,king22}.  Below, we only discuss classical algorithms for Quantum MAX-CUT, i.e. we do not discuss quantum algorithms for Quantum MAX-CUT.

For the product state Quantum MAX-CUT problem, there is a conjecturally optimal approximation algorithm.  Assuming the Unique Games Conjecture and Conjecture \ref{conj1} below (a vector-valued Borell inequality), we would then have a sharp hardness of approximation for the product state of Quantum MAX-CUT.  The vector-valued Borell inequality, Conjecture \ref{conj1}  \cite{hwang21} is a sphere-valued generalization of the Borell inequality \cite{borell85}.

The main result of this paper is a proof the vector-valued Borell Inequality (Conjecture \ref{conj1}) for all correlations $\rho$ satisfying $\abs{\rho}<.104$.  It remains a challenge to prove Conjecture \ref{conj1} for all $\abs{\rho}<1$.

Our focus in this paper is proving the inequality conjectured in \cite{hwang21} for functions $f\colon\R^{n}\to S^{2}$.  Our methods work equally well for functions taking values in $S^{k}$ for any $k\geq2$, but the values of $\rho$ for which our proof works would then depend on $k$.  Since the case $k=2$ is the only one needed for the product state Quantum MAX-CUT problem, we have therefore only focused on the case $k=2$ in this paper.

\subsection{Quantum MAX-CUT}

Below we describe the Quantum MAX-CUT problem by analogy with MAX-CUT.  When $M$ is a $2\times 2$ matrix and $j$ is a positive integer, we denote
$$M^{\otimes j}\colonequals\underbrace{M\otimes \cdots\otimes M}_{j\rm\ times}.$$
If $n$ is a positive integer and $1\leq j\leq n$, denote
$$Z_{j}\colonequals I_{2}^{\otimes (j-1)}\otimes \begin{pmatrix} 1 & 0 \\ 0 & -1\end{pmatrix} \otimes I_{2}^{\otimes (n-j)},\qquad\forall\,1\leq j\leq n,
\qquad I_{2}\colonequals\begin{pmatrix} 1 & 0 \\ 0 & 1\end{pmatrix}.$$
The (weighted) \textbf{MAX-CUT} problem can be equivalently stated as \cite{gharibian19,hwang21}: given $w\colon\{1,\ldots,n\}^{2}\to[0,\infty)$ satisfying $w_{ij}=w_{ji}$ and $w_{ii}=0$ for all $1\leq i,j\leq n$, compute the following quantity
$$\max_{u\in(\C^{2})^{\otimes n}\colon\vnorm{u}\leq1}u^{*}\Big(\sum_{i,j=1}^{n}w_{ij}(I_{2}^{\otimes n} - Z_{i}Z_{j})\Big) u.$$
Define now

$$X_{j}\colonequals I_{2}^{\otimes (j-1)}\otimes \begin{pmatrix} 0 & 1 \\ 1 & 0\end{pmatrix} \otimes I_{2}^{\otimes (n-j)},\qquad\forall\,1\leq j\leq n,$$
$$Y_{j}\colonequals I_{2}^{\otimes (j-1)}\otimes \begin{pmatrix} 0 & -\sqrt{-1} \\ \sqrt{-1} & 0\end{pmatrix} \otimes I_{2}^{\otimes (n-j)},\qquad\forall\,1\leq j\leq n.$$

The \textbf{Quantum MAX-CUT} problem is \cite{gharibian19,king22,hwang21}: given $w\colon\{1,\ldots,n\}^{2}\to[0,\infty)$ satisfying $w_{ij}=w_{ji}$ and $w_{ii}=0$ for all $1\leq i,j\leq n$, compute the following quantity
$$\max_{u\in\C^{2n}\colon \vnorm{u}\leq 1}u^{*}\Big(\sum_{i,j=1}^{n}w_{ij}(I_{2}^{\otimes n}- X_{i} X_{j} - Y_{i} Y_{j} - Z_{i} Z_{j})\Big) u.$$
The \textbf{product state of Quantum MAX-CUT} is the more restricted optimization problem of computing
$$\max_{\substack{u=u_{1}\otimes\cdots\otimes u_{n}\colon\\ u_{i}\in \C^{2},\,\vnorm{u_{i}}\leq1,\,\forall\,1\leq i\leq n}}u^{*}\Big(\sum_{i,j=1}^{n}w_{ij}(I_{2}^{\otimes n} - X_{i} X_{j} - Y_{i} Y_{j} - Z_{i} Z_{j})\Big) u.$$

\subsection{Some Notation}

For any $x=(x_{1},\ldots,x_{\sdimn}),y=(y_{1},\ldots,y_{\sdimn})\in\C^{\sdimn}$, denote
$$\langle x,y\rangle\colonequals\sum_{i=1}^{n}x_{i}\overline{y_{i}},\qquad\vnorm{x}\colonequals\langle x,x\rangle^{1/2}.$$
Denote the $(\sdimn-1)$-dimensional sphere in $\R^{\sdimn}$ as
$$S^{\sdimn-1}\colonequals\{x\in\R^{\sdimn}\colon \vnorm{x}=1\}.$$
Define the Gaussian density function in $\R^{\adimn}$ as
$$\gamma_{\adimn}(x)\colonequals(2\pi)^{-\adimn/2}e^{-\vnorm{x}^{2}/2},\qquad\forall\,x\in\R^{\adimn}.$$
When $A\subset\R^{\adimn}$ is a measurable set, denote the Gaussian measure of $A$ as $\gamma_{\adimn}(A)\colonequals\int_{A}\gamma_{\adimn}(x)\,\d x$.

\begin{definition}[\embolden{Ornstein-Uhlenbeck Operator}]\label{oudef}
Let $-1<\rho<1$.  Let $f\colon\R^{\adimn}\to[0,1]$ be measurable.  Define the \textbf{Ornstein-Uhlenbeck} operator applied to $f$ by
$$T_{\rho}f(x)\colonequals \int_{\R^{\adimn}}f(\rho x+y\sqrt{1-\rho^{2}})\gamma_{\adimn}(x)\,\d x,\qquad\forall\,x\in\R^{\adimn}.$$
\end{definition}

\begin{definition}[\embolden{Noise Stability}]\label{nsdef}
Let $-1<\rho<1$.  Let $\Omega\subset\R^{\adimn}$ be measurable.  Define the \textbf{noise stability} of $\Omega$ with correlation $\rho$, to be
$$\int_{\R^{\adimn}}1_{\Omega}(x)T_{\rho}1_{\Omega}(x)\,\gamma_{\adimn}(x)\,\d x.$$
More generally, for any bounded measurable $f\colon\R^{\adimn}\to\R^{k}$, define its noise stability with correlation $\rho$, to be
$$\int_{\R^{\adimn}}\langle f(x), T_{\rho}f(x)\rangle\gamma_{\adimn}(x)\,\d x.$$
\end{definition}

\subsection{Vector-Valued Borell Inequality}

For any positive integer $k$, denote
\begin{equation}\label{foptdef}
f_{\rm opt}(x)\colonequals\frac{x}{\vnorm{x}},\qquad\forall\,x\in\R^{k}\setminus\{0\}.
\end{equation}

\begin{conj}[{\cite{hwang21}}]\label{conj1}
Let $n\geq k$ be positive integers.  Let $f\colon \R^{n}\to S^{k-1}$ be measurable.  Then
\begin{itemize}
\item If $0<\rho<1$ and if $\int_{\R^{n}}f(x)\gamma_{n}(x)\, \d x=0$, then
$$\int_{\R^{n}}\langle f(x), T_{\rho}f(x)\rangle \gamma_{n}(x)\, \d x\leq \int_{\R^{k}}\langle f_{\rm opt}(x), T_{\rho}f_{\rm opt}(x)\rangle \gamma_{k}(x)\, \d x.$$
\item If $-1<\rho<0$, then
$$\int_{\R^{n}}\langle f(x), T_{\rho}f(x)\rangle \gamma_{n}(x)\, \d x\geq \int_{\R^{k}}\langle f_{\rm opt}(x), T_{\rho}f_{\rm opt}(x)\rangle \gamma_{k}(x)\, \d x.$$
\end{itemize}
Moreover, equality holds only when there exists a real orthogonal $k\times n$ matrix $M$ ($MM^{T}=I$) such that $f(x)=f_{\rm opt}(Mx)$ for a.e. $x\in\R^{n}$.
\end{conj}
\begin{remark}
The case $k=1$ of Conjecture \ref{conj1} is known to be true, since it is Borell's inequality \cite{borell85}.
\end{remark}

\subsection{Our Result}

\begin{theorem}[\embolden{Main}]\label{thm1}
Conjecture \ref{conj1} holds for all $n\geq k=3$ and for all $-.104<\rho<.104$.
\end{theorem}
In fact, we prove a stronger ``stable'' version of Conjecture \ref{conj1} when $n=k=3$, in \eqref{four4}, \eqref{five4} below:
\begin{itemize}
\item If $0<\rho<.104$ and if $\int_{\R^{k}}f(x)\gamma_{k}(x)\, \d x=0$, then
\begin{flalign*}
\int_{\R^{k}}\langle f(x), T_{\rho}f(x)\rangle \gamma_{k}(x)\, \d x
&\leq\int_{\R^{k}}\langle f_{\rm opt}(x), T_{\rho}f_{\rm opt}(x)\rangle \gamma_{k}(x)\, \d x\\
&\qquad\quad+(9.4\rho-.98)\int_{\R^{k}}\phi(x)\vnorm{\int_{y\in\vnorm{x}S^{k-1}}f(y) \,\d\sigma(y)}^{2}\gamma_{k}(x)\,\d x.
\end{flalign*}
\item If $-.104<\rho<0$, then
\begin{flalign*}
\int_{\R^{k}}\langle f(x), T_{\rho}f(x)\rangle \gamma_{k}(x)\, \d x
&\geq \int_{\R^{k}}\langle f_{\rm opt}(x), T_{\rho}f_{\rm opt}(x)\rangle \gamma_{k}(x)\, \d x\\
&\qquad\quad+(.98-9.4\abs{\rho})\int_{\R^{k}}\phi(x)\vnorm{\int_{y\in\vnorm{x}S^{k-1}}f(y) \,\d\sigma(y)}^{2}\gamma_{k}(x)\,\d x.
\end{flalign*}
\end{itemize}
Here $\phi(x)=1-\frac{1}{\rho \vnorm{x}}+\frac{2}{e^{2\rho \vnorm{x}}-1}$ for all $x\in\R^{k}\setminus\{0\}$ and $\d\sigma$ denotes normalized (Haar) measure on the sphere.

We can also prove Theorem \ref{thm1} when $k>3$, but with different dependence on $\rho$.  Since the $k=3$ case is the only relevant case for computational hardness of the product state of Quantum MAX-CUT, we only state a result for $k=3$.

As shown in \cite[Theorem 6.11]{hwang21} by adapting the argument of \cite{heilman20d}, the case $n=k=3$ of Theorem \ref{thm1} (when $\rho>0$) implies the case $n>k=3$ of Theorem \ref{thm1} (when $\rho>0$).  That is, the dimension of the domain can be a priori reduced to the case $n=k$, at least when $\rho>0$.  In this paper, we show that a similar statement holds in the case $\rho<0$.  This ``dimension reduction'' argument works differently in the case $\rho<0$.  For technical reasons, we need to instead deal with a bilinear version of the noise stability.  With this change, we show that the bilinear noise stability has a similar dimension reduction argument.  That is, we verify that the case $n=k=3$ of Theorem \ref{thm1} (when $\rho<0$) implies the case $n>k=3$ of Theorem \ref{thm1} (when $\rho<0$).  The bilinear inequality we prove in the case $n=k=3$ is the following.  Suppose $0<\rho<.104$ and $n=k=3$.  We then show in \eqref{five4} below that
$$\int_{\R^{k}}\langle f(x), T_{\rho}g(x)\rangle \gamma_{k}(x)\, \d x\geq -\int_{\R^{k}}\langle f_{\rm opt}(x), T_{\rho}f_{\rm opt}(x)\rangle \gamma_{k}(x)\, \d x,$$
$\forall$ $f,g\colon\R^{k}\to S^{k-1}$ when $k=3$, if $\int_{\R^{k}}f(x)\gamma_{k}(x)\,\d x=\int_{\R^{k}}g(x)\gamma_{k}(x)\,\d x$, or more generally,
\begin{flalign*}
&\int_{\R^{k}}\langle f(x), T_{\rho}g(x)\rangle \gamma_{k}(x)\, \d x
\geq -\int_{\R^{k}}\langle f_{\rm opt}(x), T_{\rho}f_{\rm opt}(x)\rangle \gamma_{k}(x)\, \d x\\
&\quad+(.98-9.4\rho)\int_{\R^{k}}\phi(x)\frac{1}{2}\Big(\vnorm{\int_{y\in\vnorm{x}S^{k-1}}f(y) \,\d\sigma(y)}^{2}
+\vnorm{\int_{y\in\vnorm{x}S^{k-1}}g(y) \,\d\sigma(y)}^{2}\Big)\gamma_{k}(x)\,\d x.
\end{flalign*}

Without considering this bilinear noise stability, \cite[Theorem 5.1]{hwang21} proved only the case $n=k$ of Conjecture \ref{conj1} when $\rho<0$.  The case $n=k$ in Conjecture \ref{conj1} does not imply the cases $n>k$, and this is why we considered the more general bilinear noise stability inequality above.  (Although one can e.g. regard a function $f\colon \R^{4}\to S^{2}$ as a function $\overline{f}\colon \R^{4}\to S^{3}$ by considering $S^{2}$ as a subset of $S^{3}$, the inequality satisfied by $\overline{f}$ proven in \cite[Theorem 5.1]{hwang21} is not a sharp inequality for $f$.)

Applications to Quantum MAX-CUT require Conjecture \ref{conj1} to hold for all $n\geq k$.  So, we have managed to prove a result that is independent of the dimension of the domain, since it holds for any $n\geq k$ when $k=3$, though not for the full range of $\rho$ parameters.  Also, the proof method of \cite{hwang21} does not directly apply in the case $\rho>0$, and attempting to prove Conjecture \ref{conj1} in that case led us towards Theorem \ref{thm1}, although the case $\rho>0$ is irrelevant for applications to Quantum MAX-CUT.

Conjecture \ref{conj1} should hold for all $-1<\rho<1$ and for all functions $n\geq k=3$.   If such a conjecture holds, then we would be able to conclude sharp hardness of approximation for the product state of Quantum MAX-CUT problem, assuming the Unique Games Conjecture.

\begin{conj}[\embolden{Sharp Hardness for Quantum MAX-CUT}, {\cite{hwang21}}]\label{conj2}
Assume that the Unique Games Conjecture is true \cite{khot02,khot10a,khot18}.  Assume Conjecture \ref{conj1} holds for all $n\geq k=3$.  Then, for any $\epsilon>0$, it is NP-hard to approximate the product state of Quantum MAX-CUT within a multiplicative factor of $\alpha_{\rm BOV}+\epsilon$.
\end{conj}

As shown in \cite{briet10,hwang21}, we have
\begin{equation}\label{abovdef}
\alpha_{\rm BOV}
\colonequals
\inf_{-1\leq\rho\leq 1}\frac{1-F^{*}(3,\rho)}{1-\rho}\approx .956.
\end{equation}
\begin{equation}\label{fstdef}
F^{*}(3,\rho)\colonequals\frac{2}{3}\Big(\frac{\Gamma((3+1)/2)}{\Gamma(3/2)}\Big)^{2}\rho\cdot\,\,_{2}F_{1}(1/2,1/2,3/2 +1, \rho^{2}),\qquad\forall\,-1<\rho<1.
\end{equation}
Here $_{2}F_{1}(\cdot,\cdot,\cdot,\cdot)$ is the Gaussian hypergeometric function.

%

The semidefinite programming algorithm of \cite{briet10} shows that Conjecture \ref{conj2} is sharp, since the polynomial time algorithm of \cite{briet10} approximates the product state Quantum MAX-CUT problem with a multiplicative factor of $\alpha_{\rm BOV}-\epsilon$, for any $\epsilon>0$.

A corollary of our main result Theorem \ref{thm1} is a weaker version of Conjecture \ref{conj2} that replaces the sharp constant $.956\ldots$ with a larger constant.

\begin{theorem}[\embolden{Unique Games Hardness for Product State Quantum MAX-CUT}]\label{thm7}
Assume that the Unique Games Conjecture is true.  Then it is NP-hard to approximate the product state of Quantum MAX-CUT within a multiplicative factor of $.9859$.
\end{theorem}

To the author's knowledge, Theorem \ref{thm7} is the only computational hardness result for the product state of Quantum MAX-CUT besides Conjecture \ref{conj2}.


\subsection{Some Notation and Definitions}

\begin{definition}[\embolden{Correlated Gaussians}]\label{gausdef}
Let $-1<\rho<1$.  Let $G_{\rho}(x,y)$ denote the joint probability density function on $\R^{\sdimn}\times\R^{\sdimn}$ such that
\begin{equation}\label{gdef}
G_{\rho}(x,y)\colonequals\frac{1}{(2\pi)^{\sdimn}(1-\rho^{2})^{\sdimn/2}}e^{\frac{-\vnorm{x}^{2}-\vnorm{y}^{2}+2\rho \langle x,y\rangle}{2(1-\rho^{2})}},\qquad\forall\,x,y\in \R^{\sdimn}.
\end{equation}

We denote $X\sim_{\rho} Y$ when $(X,Y)\in\R^{n}\times \R^{n}$ have joint probability density function $G_{\rho}$.
\end{definition}

\begin{definition}[\embolden{Correlated Random Variables on the Sphere}]\label{uvdef}
Let $-1<\rho<1$ and let $r,s>0$.  Let $G_{\rho}^{r,s}(u,v)$ denote the probability density function on $S^{\sdimn-1}\times S^{\sdimn-1}$ such that the first variable is uniform on $S^{\sdimn-1}$ and such that the second variable conditioned on the first has conditional density
$$G_{\rho}^{r,s}(v|u)\colonequals\frac{1}{z_{\rho,r,s}}e^{\frac{\rho r s\langle u,v\rangle}{1-\rho^{2}}},\qquad\forall\,v\in S^{\sdimn-1}.$$
Here $z_{\rho,r,s}$ is a normalizing constant, chosen so that $\int_{S^{\sdimn-1}}G_{\rho}^{r,s}(v|u)\d\sigma(v)=1$, where $\sigma$ denotes the uniform probability (Haar) measure on $S^{\sdimn-1}$.

We let $N_{\rho}^{r,s}$ denote the above distribution on $S^{\sdimn-1}\times S^{\sdimn-1}$ and we denote $(U,V)\sim N_{\rho}^{r,s}$ when $(U,V)\in S^{\sdimn-1}\times S^{\sdimn-1}$ have the distribution $N_{\rho}^{r,s}$.
\end{definition}

\begin{definition}[\embolden{Spherical Noise Stability}]\label{spnsdef}
Let $\rho\in(-1,1)$, $r,s>0$.  Let $f\colon S^{\sdimn-1}\to[0,1]$ be measurable.  Define $g=g_{\rho,r,s}\colon[-1,1]\to\R$ by \eqref{one0p}.  Define the smoothing operator $U_{g}$ applied to $f$ by
$$U_{g}f(x)\colonequals \int_{S^{\sdimn-1}}g(\langle x,y\rangle)f(y)\,\d\sigma(y),\qquad\forall x\in S^{\sdimn}.$$
Here $\sigma$ denotes the (normalized) Haar probability measure on $S^{\sdimn-1}$.  The \textbf{spherical noise stability} of a set $\Omega\subset S^{\sdimn-1}$ with parameters $\rho,r,s$ is
$$\int_{S^{\sdimn-1}}1_{\Omega}(x)U_{g}1_{\Omega}(x)\,\d\sigma(x).$$
\end{definition}

The spherical noise stability has a decomposition into spherical harmonics by the Funk-Hecke Formula \cite[Theorem 4.3]{hwang21}
\begin{equation}\label{nine7}
\int_{S^{n-1}}f(x)U_{g}f(x)\,\d\sigma(x)
=\Big\|\int_{S^{n-1}}f(x)\d\sigma(x)\Big\|^{2}+\sum_{d=1}^{\infty}\lambda_{d,\sdimn}^{r,s}\vnormf{\mathrm{Proj}_{d}(f)}^{2}.
\end{equation}
Here $\mathrm{Proj}_{d}(f)$ is the $L_{2}(\d\sigma)$ projection of $f$ onto spherical harmonics of degree $d$, and $\lambda_{d,n}^{r,s}$ are specified in \eqref{lamdef}

Fix $r,s>0$ and let $0<\rho<1$.  Define $g\colon[-1,1]\to\R$ by
\begin{equation}\label{one0}
g(t)=g_{\rho,r,s}(t)\colonequals\sqrt{\pi}\frac{\Gamma((n-1)/2)}{\Gamma(n/2)}\frac{e^{\frac{\rho rst}{1-\rho^{2}}}}{\int_{-1}^{1}(1-a^{2})^{\frac{\sdimn}{2}-\frac{3}{2}}e^{\frac{\rho rsa}{1-\rho^{2}}}\d a},\qquad\forall\,t\in[-1,1].
\end{equation}
Recall that, if $h\colon\R\to\R$ is continuous, then

\begin{flalign*}
\frac{1}{\mathrm{Vol}(S^{n-1})}\int_{S^{n-1}}h(y_{1})\d y
&=\frac{\mathrm{Vol}(S^{n-2})}{\mathrm{Vol}(S^{n-1})}\int_{-1}^{1}(1-t^{2})^{\frac{\sdimn}{2}-\frac{3}{2}}h(t)\d t\\
&=\frac{2\pi^{(n-1)/2}/\Gamma((n-1)/2)}{2\pi^{n/2}/\Gamma(n/2)}\int_{-1}^{1}(1-t^{2})^{\frac{\sdimn}{2}-\frac{3}{2}}h(t)\d t\\
&=\frac{1}{\sqrt{\pi}}\frac{\Gamma(n/2)}{\Gamma((n-1)/2)}\int_{-1}^{1}(1-t^{2})^{\frac{\sdimn}{2}-\frac{3}{2}}h(t)\d t.
\end{flalign*}
We have chosen the constants so that $1=\frac{1}{\mathrm{Vol}(S^{n-1})}\int_{S^{n-1}}h(y_{1})\d y$, when $h=g$.  When $h\colonequals1$, we have $\int_{-1}^{1}(1-t^{2})^{\frac{n}{2}-\frac{3}{2}}\d t=\sqrt{\pi}\frac{\Gamma((n-1)/2)}{\Gamma(n/2)}$, so 
\begin{equation}\label{one0p}
g(t)=g_{\rho,r,s}(t)\stackrel{\eqref{one0}}{=}e^{\frac{\rho rst}{1-\rho^{2}}}\cdot \frac{\int_{-1}^{1}(1-a^{2})^{\frac{n}{2}-\frac{3}{2}}\d a}{\int_{-1}^{1}(1-a^{2})^{\frac{\sdimn}{2}-\frac{3}{2}}e^{\frac{\rho rsa}{1-\rho^{2}}}\d a},\qquad\forall\,t\in[-1,1].
\end{equation}

\textbf{Notation: Rising Factorial}.  For any $x\in\R$ and for any integer $d\geq1$, we denote $(x)_{d}\colonequals \prod_{j=0}^{d-1}(x+j)$.

Let $C_{d}^{(\alpha)}\colon[-1,1]\to\R$ denote the index $\alpha$ degree $d$ Gegenbauer polynomial, which satisfies a Rodrigues formula
\cite[p. 303, 6.4.14]{andrews99}
$$(1-t^{2})^{\alpha-1/2}C_{d}^{(\alpha)}(t)=\frac{(-2)^{d}(\alpha)_{d}}{d!(d+2\alpha)_{d}}\frac{\d^{d}}{\d t^{d}}(1-t^{2})^{\alpha+d-1/2},\qquad\forall\,t\in[-1,1].$$
Letting $\alpha\colonequals\frac{n}{2}-1$, we have
\begin{equation}\label{one1}
(1-t^{2})^{\frac{n}{2}-\frac{3}{2}}C_{d}^{(\frac{n}{2}-1)}(t)=\frac{(-2)^{d}\Big(\frac{n}{2}-1\Big)_{d}}{d!(d+n-2)_{d}}\frac{\d^{d}}{\d t^{d}}(1-t^{2})^{\frac{n}{2}+d-\frac{3}{2}},\qquad\forall\,t\in[-1,1].
\end{equation}
From \cite[p. 302]{andrews99},
\begin{equation}\label{one2}
C_{d}^{(\frac{n}{2}-1)}(1)=\frac{(n-2)_{d}}{d!}.
\end{equation}
Then \cite[Corollary 4.6]{hwang21} defines
\begin{equation}\label{lamdef}
\begin{aligned}
\lambda_{d,\sdimn}^{r,s}
&\colonequals\frac{\int_{-1}^{1}\frac{C_{d}^{(\frac{n}{2}-1)}(t)}{C_{d}^{(\frac{n}{2}-1)}(1)}(1-t^{2})^{\frac{n}{2}-\frac{3}{2}} g(t)\d t}{\int_{-1}^{1}(1-t^{2})^{\frac{n}{2}-\frac{3}{2}}\d t}
\stackrel{\eqref{one0p}}{=}\frac{\int_{-1}^{1}\frac{C_{d}^{(\frac{n}{2}-1)}(t)}{C_{d}^{(\frac{n}{2}-1)}(1)}(1-t^{2})^{\frac{n}{2}-\frac{3}{2}} e^{\frac{\rho rst}{1-\rho^{2}}}\d t}
{\int_{-1}^{1}(1-t^{2})^{\frac{n}{2}-\frac{3}{2}}e^{\frac{\rho rst}{1-\rho^{2}}}\d t}\\
&\stackrel{\eqref{one1}\wedge\eqref{one2}}{=}
\frac{(-2)^{d}(\frac{n}{2}-1)_{d}}{d!(d+n-2)_{d}}\frac{d!}{(n-2)_{d}}
\frac{\int_{-1}^{1} \Big[\frac{\d^{d}}{\d t^{d}}(1-t^{2})^{\frac{n}{2}+d-\frac{3}{2}}\Big]e^{\frac{\rho rst}{1-\rho^{2}}}\d t}
{\int_{-1}^{1}(1-t^{2})^{\frac{n}{2}-\frac{3}{2}}e^{\frac{\rho rst}{1-\rho^{2}}}\d t}.
\end{aligned}
\end{equation}
Integrating by parts $d$ times,
\begin{flalign*}
\lambda_{d,\sdimn}^{r,s}
&=\Big(\frac{\rho rs}{1-\rho^{2}}\Big)^{d}
\frac{(-2)^{d}(\frac{n}{2}-1)_{d}}{(n-2)_{2d}}
\frac{(-1)^{d}\int_{-1}^{1} (1-t^{2})^{\frac{n}{2}+d-\frac{3}{2}}e^{\frac{\rho rst}{1-\rho^{2}}}\d t}
{\int_{-1}^{1}(1-t^{2})^{\frac{n}{2}-\frac{3}{2}}e^{\frac{\rho rst}{1-\rho^{2}}}\d t}.
\end{flalign*}

In the case $d=1$ we have
\begin{equation}\label{one6}
\begin{aligned}
\lambda_{1,\sdimn}^{r,s}
&=\Big(\frac{\rho rs}{1-\rho^{2}}\Big)
\frac{(2)(\frac{n}{2}-1)}{(n-1)}\frac{1}{(n-2)}
\frac{\int_{-1}^{1} (1-t^{2})^{\frac{n}{2}-\frac{1}{2}}e^{\frac{\rho rst}{1-\rho^{2}}}\d t}
{\int_{-1}^{1}(1-t^{2})^{\frac{n}{2}-\frac{3}{2}}e^{\frac{\rho rst}{1-\rho^{2}}}\d t}\\
&=\Big(\frac{\rho rs}{1-\rho^{2}}\Big)\Big(\frac{1}{n-1}\Big)
\frac{\int_{-1}^{1} (1-t^{2})^{\frac{n}{2}-\frac{1}{2}}e^{\frac{\rho rst}{1-\rho^{2}}}\d t}
{\int_{-1}^{1}(1-t^{2})^{\frac{n}{2}-\frac{3}{2}}e^{\frac{\rho rst}{1-\rho^{2}}}\d t}.
\end{aligned}
\end{equation}

The following Lemma follows from the Cauchy-Schwarz inequality, and from the spherical harmonic decomposition \eqref{nine7}.

\begin{lemma}[{\cite[Lemma 5.4]{hwang21}}]\label{lemma1}
Let $r,s>0$.  Let $\rho>0$.  Let $f_{r},f_{s}\colon S^{\sdimn-1}\to S^{\sdimn-1}$.  Then
$$\E_{(U,V)\sim N_{\rho}^{r,s}}\langle f_{r}(U),f_{s}(V)\rangle\leq\langle \E f_{r},\E f_{s}\rangle+\lambda_{1,\sdimn}^{r,s}\sqrt{\E\vnorm{f_{r}-\E f_{r}}^{2}}\sqrt{\E\vnorm{f_{s}-\E f_{s}}^{2}}.$$
Equality holds only when $f_{r}(x)=f_{s}(x)=Mx$ for all $x\in S^{\sdimn-1}$, where $M$ is an $n\times n$ real orthogonal matrix.
$$\E_{(U,V)\sim N_{\rho}^{r,s}}\langle f_{r}(U),f_{s}(V)\rangle\geq\langle \E f_{r},\E f_{s}\rangle-\lambda_{1,\sdimn}^{r,s}\sqrt{\E\vnorm{f_{r}-\E f_{r}}^{2}}\sqrt{\E\vnorm{f_{s}-\E f_{s}}^{2}}.$$
Equality holds only when $f_{r}(x)=-f_{s}(x)=Mx$ for all $x\in S^{\sdimn-1}$, where $M$ is an $n\times n$ real orthogonal matrix.
\end{lemma}

\subsection{Expected Value Notation}\label{enote}

\begin{itemize}
\item $\E$ with no subscript denotes expected value on a sphere with respect to the uniform (Haar) probability measure.
\item $\E_{(U,V)\sim N_{\rho}^{r,s}}$ denotes expected value with respect to $(U,V)$ from Definition \ref{uvdef}.
\item $\underset{X\sim_{\rho}Y}{\E}$ denotes expected value with respect to $(X,Y)$ from Definition \ref{gausdef}.
\item $\E_{R,S}$ denotes expected value with respect to $R,S$ where $R=\vnorm{X},S=\vnorm{Y}$, and $X,Y$ are two standard $\rho$-correlated Gaussians, as in Definition \ref{gausdef}.
\item $\E_{\gamma}$ denotes expected value with respect to the Gaussian density $\gamma_{n}$.
\end{itemize}

\section{Preliminaries: Quadratic Case}

Our first step towards proving Theorem \ref{thm1} with $\rho>0$ will be modifying Lemma \ref{lemma1}.  Lemmas \ref{lemma2} and \ref{lemma3} below demonstrate that optimizing the noise stability $\underset{X\sim_{\rho}Y}{\E}\langle f(X),f(Y)\rangle$ involves an interplay between $\E f_{\vnorm{X}}$ having norm $0$ or norm $1$.

\begin{lemma}\label{lemma2}
Let $f\colon \R^{\sdimn}\to S^{\sdimn-1}$ be continuous.  For any $r>0$, denote $f_{r}\colonequals f|_{rS^{\sdimn-1}}$ and denote $\E f_{r}$ as the expected value of $f_{r}$ on $rS^{\sdimn-1}$ with respect to the uniform probability (Haar) measure on $S^{\sdimn-1}$.  Then

$$\underset{X\sim_{\rho}Y}{\E}\langle f(X),f(Y)\rangle
\leq\underset{X\sim_{\rho}Y}{\E}\Big(\langle \E f_{\vnorm{X}}, \E f_{\vnorm{Y}}\rangle+\lambda_{1,\sdimn}^{\vnorm{X},\vnorm{Y}}\sqrt{1-\vnorm{\E f_{\vnorm{X}}}^{2}}\sqrt{1-\vnorm{\E f_{\vnorm{Y}}}^{2}}\Big).$$
Equality holds only when $f(x)=Mx/\vnorm{Mx}$ for a.e. $x\in S^{\sdimn-1}$, where $M$ is an $n\times n$ real orthogonal matrix.
\end{lemma}
\begin{proof}
We first write (recalling the notation of Section \ref{enote})
$$\E_{X\sim_{\rho}Y}\langle f(X),f(Y)\rangle
=\E_{R,S}\E_{(U,V)\sim N_{\rho}^{R,S}}\langle f_{R}(U),f_{S}(V)\rangle.$$
Applying Lemma \ref{lemma1} and averaging over $R,S$, and also using that $f$ takes values in $S^{\sdimn-1}$, so $\E\vnorm{f_{S}-\E f_{S}}^{2}=1+\vnorm{\E f_{S}}^{2}-2\langle\E f_{S},\E f_{S}\rangle=1-\vnorm{\E f_{S}}^{2}$,
\begin{flalign*}
&\underset{R,S}{\E}\underset{(U,V)\sim N_{\rho}^{R,S}}{\E}\langle f_{R}(U),f_{S}(V)\rangle\\
&\qquad\leq\underset{R,S}{\E}\langle \E f_{R},\E f_{S}\rangle+\underset{R,S}{\E}\lambda_{1,\sdimn}^{R,S}\sqrt{\E\vnorm{f_{R}-\E f_{R}}^{2}}\sqrt{\E\vnorm{f_{S}-\E f_{S}}^{2}}\\
&\qquad=\underset{R,S}{\E}\langle \E f_{R},\E f_{S}\rangle+\underset{R,S}{\E}\lambda_{1,\sdimn}^{R,S}\sqrt{\E(1-\vnorm{\E f_{R}}^{2})}\sqrt{\E(1-\vnorm{\E f_{S}}^{2})}\\
&\qquad=\underset{X\sim_{\rho}Y}{\E}\Big(\langle \E f_{\vnorm{X}}, \E f_{\vnorm{Y}}\rangle+\lambda_{1,\sdimn}^{\vnorm{X},\vnorm{Y}}\sqrt{1-\vnorm{\E f_{\vnorm{X}}}^{2}}\sqrt{1-\vnorm{\E f_{\vnorm{Y}}}^{2}}\Big).
\end{flalign*}
The equality case follows from the equality case of Lemma \ref{lemma1}
\end{proof}

Using again the notation of Lemma \ref{lemma2}, we have
\begin{lemma}\label{lemma3}
Let $f\colon \R^{\sdimn}\to S^{\sdimn-1}$.  Then
$$
\underset{X\sim_{\rho}Y}{\E}\langle f(X),f(Y)\rangle
\leq\underset{X\sim_{\rho}Y}{\E}\Big(\lambda_{1,\sdimn}^{\vnorm{X},\vnorm{Y}}
+\langle \E f_{\vnorm{X}}, \E f_{\vnorm{Y}}\rangle
-\vnorm{\E f_{\vnorm{X}}}^{2}\lambda_{1,\sdimn}^{\vnorm{X},\vnorm{Y}}\Big).
$$
Equality holds only when $f(x)=Mx/\vnorm{Mx}$ for a.e. $x\in S^{\sdimn-1}$, where $M$ is an $n\times n$ real orthogonal matrix.
\end{lemma}
\begin{proof}
Applying the inequality $\abs{ab}\leq (1/2)(a^{2}+b^{2})$ $\forall$ $a,b\in\R$, we have
\begin{flalign*}
&\underset{X\sim_{\rho}Y}{\E}\lambda_{1,\sdimn}^{\vnorm{X},\vnorm{Y}}\sqrt{1-\vnorm{\E f_{\vnorm{X}}}^{2}}\sqrt{1-\vnorm{\E f_{\vnorm{Y}}}^{2}}\\
&\qquad\qquad\leq\underset{X\sim_{\rho}Y}{\E}\lambda_{1,\sdimn}^{\vnorm{X},\vnorm{Y}}(1-\vnorm{\E f_{\vnorm{X}}}^{2})\\
&\qquad\qquad=\underset{X\sim_{\rho}Y}{\E}\lambda_{1,\sdimn}^{\vnorm{X},\vnorm{Y}}-\underset{X\sim_{\rho}Y}{\E}\lambda_{1,\sdimn}^{\vnorm{X},\vnorm{Y}}\vnorm{\E f_{\vnorm{X}}}^{2}.
\end{flalign*}
Combining with Lemma \ref{lemma2} concludes the proof.
\end{proof}

\section{Preliminaries: Bilinear Case}

As in the previous section, our first step towards proving Theorem \ref{thm1} with $\rho<0$ will be modifying Lemma \ref{lemma1}.  Lemmas \ref{lemma6} and \ref{lemma7} below demonstrate that optimizing the bilinear noise stability $\underset{X\sim_{\rho}Y}{\E}\langle f(X),g(Y)\rangle$ involves an interplay between $\E f_{\vnorm{X}}$ having norm $0$ or norm $1$, and similarly for $\E g_{\vnorm{Y}}$.

\begin{lemma}\label{lemma6}
Let $f,g\colon \R^{\sdimn}\to S^{\sdimn-1}$ be continuous.  For any $r>0$, denote $f_{r}\colonequals f|_{rS^{\sdimn-1}}$ and denote $\E f_{r}$ as the expected value of $f_{r}$ on $rS^{\sdimn-1}$ with respect to the uniform probability (Haar) measure on $S^{\sdimn-1}$.  Then

$$\underset{X\sim_{\rho}Y}{\E}\langle f(X), g(Y)\rangle
\geq\underset{X\sim_{\rho}Y}{\E}\Big(\langle \E f_{\vnorm{X}}, \E g_{\vnorm{Y}}\rangle-\lambda_{1,\sdimn}^{\vnorm{X},\vnorm{Y}}\sqrt{1-\vnorm{\E f_{\vnorm{X}}}^{2}}\sqrt{1-\vnorm{\E g_{\vnorm{Y}}}^{2}}\Big).$$
Equality holds only when $f(x)=-g(x)=Mx/\vnorm{Mx}$ for a.e. $x\in S^{\sdimn-1}$, where $M$ is an $n\times n$ real orthogonal matrix.
\end{lemma}
\begin{proof}
We first write (recalling the notation of Section \ref{enote})
$$\E_{X\sim_{\rho}Y}\langle f(X),g(Y)\rangle
=\E_{R,S}\E_{(U,V)\sim N_{\rho}^{R,S}}\langle f_{R}(U),g_{S}(V)\rangle.$$
Applying Lemma \ref{lemma1} and averaging over $R,S$, and also using that $f$ takes values in $S^{\sdimn-1}$, so $\E\vnorm{f_{S}-\E f_{S}}^{2}=1+\vnorm{\E f_{S}}^{2}-2\langle\E f_{S},\E f_{S}\rangle=1-\vnorm{\E f_{S}}^{2}$, and similarly for $g$,
\begin{flalign*}
&\underset{R,S}{\E}\underset{(U,V)\sim N_{\rho}^{R,S}}{\E}\langle f_{R}(U),g_{S}(V)\rangle\\
&\qquad\geq\underset{R,S}{\E}\langle \E f_{R},\E g_{S}\rangle-\underset{R,S}{\E}\lambda_{1,\sdimn}^{R,S}\sqrt{\E\vnorm{f_{R}-\E f_{R}}^{2}}\sqrt{\E\vnorm{g_{S}-\E g_{S}}^{2}}\\
&\qquad=\underset{R,S}{\E}\langle \E f_{R},\E g_{S}\rangle-\underset{R,S}{\E}\lambda_{1,\sdimn}^{R,S}\sqrt{\E(1-\vnorm{\E f_{R}}^{2})}\sqrt{\E(1-\vnorm{\E g_{S}}^{2})}\\
&\qquad=\underset{X\sim_{\rho}Y}{\E}\Big(\langle \E f_{\vnorm{X}}, \E g_{\vnorm{Y}}\rangle+\lambda_{1,\sdimn}^{\vnorm{X},\vnorm{Y}}\sqrt{1-\vnorm{\E f_{\vnorm{X}}}^{2}}\sqrt{1-\vnorm{\E g_{\vnorm{Y}}}^{2}}\Big).
\end{flalign*}
The equality case follows from the equality case of Lemma \ref{lemma1}.
\end{proof}

Using again the notation of Lemma \ref{lemma6}.  Then
\begin{lemma}\label{lemma7}
Let $f,g\colon \R^{\sdimn}\to S^{\sdimn-1}$.  Then
$$
\underset{X\sim_{\rho}Y}{\E}\langle f(X),g(Y)\rangle
\geq\underset{X\sim_{\rho}Y}{\E}\Big(-\lambda_{1,\sdimn}^{\vnorm{X},\vnorm{Y}}
+\langle \E f_{\vnorm{X}}, \E g_{\vnorm{Y}}\rangle
+\frac{1}{2}(\vnorm{\E f_{\vnorm{X}}}^{2}+\vnorm{\E g_{\vnorm{X}}}^{2})\lambda_{1,\sdimn}^{\vnorm{X},\vnorm{Y}}\Big).
$$
Equality holds only when $f(x)=-g(x)=Mx/\vnorm{Mx}$ for a.e. $x\in S^{\sdimn-1}$, where $M$ is an $n\times n$ real orthogonal matrix.
\end{lemma}
\begin{proof}
Applying the inequality $\abs{ab}\leq (1/2)(a^{2}+b^{2})$ $\forall$ $a,b\in\R$, we have
\begin{flalign*}
&-\underset{X\sim_{\rho}Y}{\E}\lambda_{1,\sdimn}^{\vnorm{X},\vnorm{Y}}\sqrt{1-\vnorm{\E f_{\vnorm{X}}}^{2}}\sqrt{1-\vnorm{\E g_{\vnorm{Y}}}^{2}}\\
&\qquad\geq-\underset{X\sim_{\rho}Y}{\E}\lambda_{1,\sdimn}^{\vnorm{X},\vnorm{Y}}(1-[\vnorm{\E f_{\vnorm{X}}}^{2}+\vnorm{\E g_{\vnorm{X}}}^{2}]/2)\\
&\qquad=-\underset{X\sim_{\rho}Y}{\E}\lambda_{1,\sdimn}^{\vnorm{X},\vnorm{Y}}+\underset{X\sim_{\rho}Y}{\E}\lambda_{1,\sdimn}^{\vnorm{X},\vnorm{Y}}(\vnorm{\E f_{\vnorm{X}}}^{2}+\vnorm{\E g_{\vnorm{X}}}^{2})/2.
\end{flalign*}
Combining with Lemma \ref{lemma6} concludes the proof.
\end{proof}

\section{Eigenvalue Bounds}

In this section, we derive some bounds on the first eigenvalue $\lambda_{1,n}^{r,s}$ as defined in \eqref{nine7}.  From \eqref{one6}, we therefore need to control the function

\begin{equation}\label{two1}
\frac{\int_{-1}^{1} (1-t^{2})^{\frac{n}{2}-\frac{1}{2}}e^{\frac{\rho rst}{1-\rho^{2}}}\d t}
{\int_{-1}^{1}(1-t^{2})^{\frac{n}{2}-\frac{3}{2}}e^{\frac{\rho rst}{1-\rho^{2}}}\d t}.
\end{equation}

This is a ratio of modified Bessel functions of the first kind.  To recall their definition, first recall the Bessel function $J_{\alpha}$ of the first kind of order $\alpha\geq0$ is \cite[p. 204]{andrews99}
$$J_{\alpha}(x)\colonequals\frac{1}{\sqrt{\pi}\Gamma(\alpha+1/2)}(x/2)^{\alpha}\int_{-1}^{1}e^{ixt}(1-t^{2})^{\alpha-1/2}\d t,\qquad\forall\,x\in\R.$$
The modified Bessel function $I_{\alpha}$ of the first kind of order $\alpha\geq0$ is then \cite[p. 222]{andrews99}
\begin{flalign*}
I_{\alpha}(x)
\colonequals i^{-\alpha}J_{\alpha}(ix)
&=\frac{1}{\sqrt{\pi}\Gamma(\alpha+1/2)}(x/2)^{\alpha}\int_{-1}^{1}e^{-xt}(1-t^{2})^{\alpha-1/2}\d t\\
&=\frac{1}{\sqrt{\pi}\Gamma(\alpha+1/2)}(x/2)^{\alpha}\int_{-1}^{1}e^{xt}(1-t^{2})^{\alpha-1/2}\d t,\qquad\forall\,x\in\R.
\end{flalign*}
$$
\frac{I_{\alpha+1}(a)}{I_{\alpha}(a)}
=\frac{(a/2)}{\alpha+1/2}\frac{\int_{-1}^{1}e^{at}(1-t^{2})^{\alpha+1/2}\d t}{\int_{-1}^{1}e^{at}(1-t^{2})^{\alpha-1/2}\d t},\qquad\forall\,a\in\R,\,\forall\,\alpha\geq0.
$$

So, setting $\alpha=(n/2)-1$ here, the ratio from \eqref{two1} is equal to 
$$\frac{n-1}{2}\Big(\frac{2(1-\rho^{2})}{\rho rs}\Big)\frac{I_{n/2}(\rho rs/[1-\rho^{2}])}{I_{(n/2)-1}(\rho rs/[1-\rho^{2}])}.$$
Combining with \eqref{one6}, we have
\begin{equation}\label{lameq}
\lambda_{1,\sdimn}^{r,s}
=\frac{\rho rs}{1-\rho^{2}}\frac{1}{n-1}\frac{n-1}{2}\Big(\frac{2(1-\rho^{2})}{\rho rs}\Big)\frac{I_{n/2}(\rho rs/[1-\rho^{2}])}{I_{(n/2) -1}(\rho rs/[1-\rho^{2}])}
=\frac{I_{n/2}(\rho rs/[1-\rho^{2}])}{I_{(n/2)-1}(\rho rs/[1-\rho^{2}])}.
\end{equation}

We have \cite[p. 241]{amos74}
\begin{equation}\label{two2}
\frac{a}{\alpha+1+\sqrt{(\alpha+1)^{2}+a^{2}}}\leq\frac{I_{\alpha+1}(a)}{I_{\alpha}(a)}\leq\frac{a}{\alpha+\sqrt{\alpha^{2}+a^{2}}},\qquad\forall\,a,\alpha\geq0.
\end{equation}
Therefore,
\begin{equation}\label{two3}
\lambda_{1,\sdimn}^{r,s}\geq\frac{\rho rs/[1-\rho^{2}]}{n/2+\sqrt{(n/2)^{2}+[\rho rs/[1-\rho^{2}]]^{2}}},\qquad\forall\,r,s>0,\,\forall\,\rho\in(0,1).
\end{equation}

In order to use our lower bounds on $\lambda_{1,\sdimn}^{r,s}$ in Lemma \ref{lemma3}, we must average $\lambda_{1,\sdimn}^{r,s}$ over $r>0$.  Let $a>0$.  Then

$$\int_{0}^{\infty}\frac{r^{3}e^{-r^{2}/2a^{2}}}{1+\sqrt{1+r^{2}}}\d r=e^{1/2a^{2}}a^{3}\int_{1/a}^{\infty}e^{-t^{2}/2}\d t.$$

Setting $n=3$ in \eqref{lameq}, we get

\begin{equation}\label{two9}
\begin{aligned}
&\int_{0}^{\infty}r^{n-1}\frac{I_{n/2}(ar)}{I_{(n/2)-1}(ar)}e^{-r^{2}/2}\d r
\stackrel{\eqref{two2}}{\geq} \int_{0}^{\infty}r^{n-1}\frac{ar}{n/2+\sqrt{(n/2)^{2}+a^{2}r^{2}}}e^{-r^{2}/2}\d r\\
&= \int_{0}^{\infty}r^{n-1}\frac{2ar/n}{1+\sqrt{1+4a^{2}r^{2}/n^{2}}}e^{-r^{2}/2}\d r
=(n/2a)\int_{0}^{\infty}(nr/[2a])^{n-1}\frac{r}{1+\sqrt{1+r^{2}}}e^{-n^{2}r^{2}/8a^{2}}\d r\\
&=(3/2a)^{3}\int_{0}^{\infty}\frac{r^{3}}{1+\sqrt{1+r^{2}}}e^{-n^{2}r^{2}/8a^{2}}\d r
=(3/2a)^{3}e^{9/8a^{2}}(2/3)^{3}a^{3}\int_{3/2a}^{\infty}e^{-t^{2}/2}\d t\\
&=e^{9/8a^{2}}\int_{3/2a}^{\infty}e^{-t^{2}/2}\d t.
\end{aligned}
\end{equation}

\subsection{First Eigenvalue, Dimension 3}

When $d=1$ and $n=3$, we have an explicit expression for $\lambda_{1,\sdimn}^{r,s}$.
\begin{equation}\label{one9}
\begin{aligned}
\lambda_{1,\sdimn}^{r,s}
&\stackrel{\eqref{lameq}}{=}\Big(\frac{\rho rs}{1-\rho^{2}}\Big)\Big(\frac{1}{2}\Big)
\frac{\int_{-1}^{1} (1-t^{2})e^{\frac{\rho rst}{1-\rho^{2}}}\d t}
{\int_{-1}^{1}e^{\frac{\rho rst}{1-\rho^{2}}}\d t}
=1-\frac{1}{\Big(\frac{\rho rs}{1-\rho^{2}}\Big)}+\frac{2}{e^{\frac{2\rho rs}{1-\rho^{2}}}-1},
\end{aligned}
\end{equation}
for all $r,s>0$, and for all $0<\rho<1$.  Here we used, with $a=\rho rs/(1-\rho^{2})$,
$$\int_{-1}^{1}e^{at}\d t=\frac{1}{a}[e^{a}-e^{-a}].$$
\begin{flalign*}
\int_{-1}^{1}(1-t^{2})e^{at}\d t
&=\frac{1}{a}\int_{-1}^{1}(1-t^{2})\frac{\d}{\d t}e^{at}\d t
=\frac{1}{a}\int_{-1}^{1}2te^{at}\d t\\
&=\frac{1}{a^{2}}\int_{-1}^{1}2t\frac{\d}{\d t}e^{at}\d t
=\frac{1}{a^{2}}\Big[2te^{at}|_{t=-1}^{t=1}-\int_{-1}^{1}2e^{at}\d t\Big]\\
&=\frac{1}{a^{2}}\Big[2(e^{a}+e^{-a})-\frac{2}{a}[e^{a}-e^{-a}]\Big].\\
\end{flalign*}
$$
\frac{\int_{-1}^{1}(1-t^{2})e^{at}\d t}{\int_{-1}^{1}e^{at}\d t}
=\frac{1}{a}\Big(2\frac{e^{a}+e^{-a}}{e^{a}-e^{-a}}-\frac{2}{a}\Big)
=\frac{1}{a}\Big(2+4\frac{e^{-a}}{e^{a}-e^{-a}}-\frac{2}{a}\Big)
=\frac{1}{a}\Big(2+4\frac{1}{e^{2a}-1}-\frac{2}{a}\Big).
$$

\section{Change of Measure}

The following inequality allows us to show that the second term in Lemma \ref{lemma3} is smaller than the last term.  The proof amounts to an elementary truncated heat kernel bound, though with a change of measure using the $\lambda_{1,n}^{r,s}$ term from \eqref{one6} (i.e. \eqref{one9}).

\begin{lemma}\label{lemma28}
Let $f\colon\R^{3}\to S^{2}$ be a radial function (for any $r>0$, the function $f|_{rS^{2}}$ is constant).  Let $\phi\colon\R^{3}\to(0,\infty)$.  Then
\begin{flalign*}
&\abs{\int_{\R^{3}}\langle f(x)-\E_{\gamma}f,\, T_{\rho}[f-\E_{\gamma}f](x)\rangle\gamma_{3}(x)\,\d x}\\
&\leq \int_{\R^{3}}\phi(x)\vnorm{f(x)}^{2}\gamma_{3}(x)\,\d x \sum_{d\geq2\colon d\,\mathrm{even}}\rho^{d}\sum_{k\in(2\N)^{2}\colon \vnorm{k}_{1}=d}\int_{\R^{3}} \abs{\sqrt{k!}h_{k}(x)\frac{1}{\sqrt{\phi(x)}}}^{2}\gamma_{3}(x)\,\d x.
\end{flalign*}
In particular, if $\phi(x)=1-\frac{1}{\rho\vnorm{x}}+\frac{2}{e^{2\rho\vnorm{x}}-1}$ for all $x\in\R^{3}$, and if $0<\rho<1/9$, then
$$
\abs{\int_{\R^{3}}\langle f(x)-\E_{\gamma}f,\, T_{\rho}[f-\E_{\gamma}f](x)\rangle\gamma_{3}(x)\,\d x}
\leq(9.4\rho)\int_{\R^{3}}\phi(x)\vnorm{f(x)}^{2}\gamma_{3}(x)\d x.
$$
\end{lemma}
\begin{proof}
Let $h_{0},h_{1},\ldots\colon\R\to\R$ be the Hermite polynomials with $h_{m}(x)=\sum_{k=0}^{\lfloor m/2\rfloor}\frac{x^{m-2k}(-1)^{k}2^{-k}}{k!(m-2k)!}$ for all integers $m\geq0$.  It is well known \cite{heilman12} that $\{\sqrt{m!}h_{m}\}_{m\geq0}$ is an orthonormal basis of the Hilbert space of functions $\R\to\R$ equipped with the inner product $\langle g,h\rangle\colonequals\int_{\R}g(x)h(x)\gamma_{3}(x)\,\d x$.  For any $k\in\N^{2}$, define $k!\colonequals k_{1}!\cdot k_{2}!$, and define $\vnorm{k}_{1}\colonequals\abs{k_{1}}+\abs{k_{2}}$.
\begin{flalign*}
&\int_{\R^{3}}\langle f(x)-\E_{\gamma}f,\, T_{\rho}[f-\E_{\gamma}f](x)\rangle\gamma_{3}(x)\,\d x\\
&=\sum_{d\geq2\colon d\,\mathrm{even}}\rho^{d}\sum_{k\in(2\N)^{2}\colon \vnorm{k}_{1}=d}\vnorm{\int_{\R^{3}} \sqrt{k!}h_{k}(x) (f(x)-\E_{\gamma}f)\gamma_{3}(x)\,\d x}^{2}\\
&=\sum_{d\geq2\colon d\,\mathrm{even}}\rho^{d}\sum_{k\in(2\N)^{2}\colon \vnorm{k}_{1}=d}\vnorm{\int_{\R^{3}} \sqrt{k!}h_{k}(x) f(x)\gamma_{3}(x)\,\d x}^{2}\\
&=\sum_{d\geq2\colon d\,\mathrm{even}}\rho^{d}\sum_{k\in(2\N)^{2}\colon \vnorm{k}_{1}=d}\vnorm{\int_{\R^{3}} \sqrt{k!}h_{k}(x)\frac{1}{\sqrt{\phi(x)}} \sqrt{\phi(x)}f(x)\gamma_{3}(x)\,\d x}^{2}\\
&\leq\sum_{d\geq2\colon d\,\mathrm{even}}\rho^{d}\sum_{\substack{k\in(2\N)^{2}\colon\\ \vnorm{k}_{1}=d}}\int_{\R^{3}} \abs{\sqrt{k!}h_{k}(x)\frac{1}{\sqrt{\phi(x)}}}^{2}\gamma_{3}(x)\,\d x
\cdot\int_{\R^{3}}\phi(x)\vnorm{f(x)}^{2}\gamma_{3}(x)\,\d x\\
&=\int_{\R^{3}}\phi(x)\vnorm{f(x)}^{2}\gamma_{3}(x)\,\d x\cdot \sum_{d\geq2\colon d\,\mathrm{even}}\rho^{d}\sum_{\substack{k\in(2\N)^{2}\colon\\ \vnorm{k}_{1}=d}}\int_{\R^{3}} \abs{\sqrt{k!}h_{k}(x)\frac{1}{\sqrt{\phi(x)}}}^{2}\gamma_{3}(x)\,\d x.
\end{flalign*}

Using the inequality
$$\frac{r^{2}}{\phi(r)}\leq \frac{3}{\rho}r+r^{2},\qquad\forall\,r>0,$$
which can e.g. be verified in Matlab
\begin{verbatim}
rho=.1;
r=linspace(.1,20,1000);
rsqphi = r.^2 ./ (1 - 1./(rho*r) + 2./(exp(2*rho*r) -1));
plot(r, rsqphi , r, 3*r/rho +  r.^2);
legend('r^2 / phi','upper bound');
if sum( rsqphi - (3*r/rho + r.^2)>0)==0, fprintf('Verified\r'), end
\end{verbatim}

\begin{flalign*}
&\sum_{d\geq2\colon d\,\mathrm{even}}\rho^{d}\sum_{k\in(2\N)^{2}\colon \vnorm{k}_{1}=d}\int_{\R^{3}} \abs{\sqrt{k!}h_{k}(x)\frac{1}{\sqrt{\phi(x)}}}^{2}\gamma_{3}(x)\,\d x\\
&\qquad=\int_{\R^{3}}\frac{1}{\phi(x)}\Big[(2\pi)^{3/2}\frac{G_{\rho}(x,x)+G_{\rho}(x,-x)}{2e^{-\vnorm{x}^{2}/2}} - \gamma_{3}(x)\Big]\, \d x\\
&\qquad=\int_{\R^{3}}\frac{1}{\phi(x)}\Big[\frac{1}{(2\pi)^{3/2}}\frac{1}{(1-\rho^{2})^{3/2}} e^{-\frac{\vnorm{x}^{2}}{1-\rho^{2}}}
\frac{e^{\frac{\rho\vnorm{x}^{2}}{1-\rho^{2}}}+e^{-\frac{\rho\vnorm{x}^{2}}{1-\rho^{2}}}}{2e^{-r^{2}/2}} - \gamma_{3}(x)\Big]\, \d x\\
&\qquad=\sqrt{\frac{2}{\pi}}\int_{r=0}^{\infty}\Big(\frac{3}{\rho}r+r^{2}\Big)\frac{1}{2(1-\rho^{2})^{3/2}}\\
&\qquad\qquad\qquad\qquad\cdot\Big(-2(1-\rho^{2})^{3/2}e^{-r^{2}/2}+e^{-r^{2}\frac{1-\rho}{1-\rho^{2}}}e^{r^{2}/2}+e^{-r^{2}\frac{1+\rho}{1-\rho^{2}}}e^{r^{2}/2}\Big)\, \d r\\
&\qquad=\sqrt{\frac{2}{\pi}}\int_{r=0}^{\infty}\Big(\frac{3}{\rho}r+r^{2}\Big)\frac{1}{2(1-\rho^{2})^{3/2}}\\
&\qquad\qquad\qquad\qquad\cdot\Big(-2(1-\rho^{2})^{3/2}e^{-r^{2}/2}+e^{-r^{2}\left(\frac{1}{1+\rho}-\frac{1}{2}\right)}+e^{-r^{2}\left(\frac{1}{1-\rho}-\frac{1}{2}\right)}\Big)\, \d r\\
&\qquad=\sqrt{\frac{2}{\pi}}\frac{1}{2(1-\rho^{2})^{3/2}}\frac{3}{\rho}\Big(-2(1-\rho^{2})^{3/2}+\frac{1}{2[\frac{1}{1+\rho}-\frac{1}{2}]}+ \frac{1}{2[\frac{1}{1-\rho}-\frac{1}{2}]} \Big)\\
&\qquad\qquad+ \sqrt{\frac{2}{\pi}}\frac{1}{(1-\rho^{2})^{3/2}}\frac{\sqrt{\pi}}{8}\Big(-2\cdot 2^{3/2}(1-\rho^{2})^{3/2}+\frac{1}{[\frac{1}{1+\rho}-\frac{1}{2}]^{3/2}}+\frac{1}{[\frac{1}{1-\rho}-\frac{1}{2}]^{3/2}}\Big).
\end{flalign*}

When $\rho<1/9$, this quantity is upper bounded by $9.4\rho$.  

\end{proof}

For the negative correlation case of the main result, we require a bilinear version of Lemma \ref{lemma28} above.  Lemma \ref{lemma29} follows from Lemma \ref{lemma28} and the Cauchy-Schwarz inequality.
\begin{lemma}\label{lemma29}
Let $f,g\colon\R^{3}\to S^{2}$ be radial functions (for any $r>0$, the function $f|_{rS^{2}}$ is constant).  Let $\phi(x)=1-\frac{1}{\rho\vnorm{x}}+\frac{2}{e^{2\rho\vnorm{x}}-1}$ for all $x\in\R^{3}$.  If $0<\rho<1/9$, then
$$
\abs{\int_{\R^{3}}\langle f(x)-\E_{\gamma}f,\, T_{\rho}[g-\E_{\gamma}g](x)\rangle\gamma_{3}(x)\,\d x}
\leq(9.4\rho)\int_{\R^{3}}\phi(x)\frac{1}{2}[\vnorm{f(x)}^{2}+\vnorm{g(x)}^{2}]\gamma_{3}(x)\d x.
$$
\end{lemma}

\section{Proof of Main Theorem: Quadratic Case}

\begin{proof}[Proof of Theorem \ref{thm1} when $\rho>0$]
The dimension reduction Theorem \cite[Theorem 6.11]{hwang21} implies that we may assume $n=k=3$.  From Lemma \ref{lemma3}
\begin{equation}\label{four0}
\underset{X\sim_{\rho}Y}{\E}\langle f(X),f(Y)\rangle
-\underset{X\sim_{\rho}Y}{\E}\lambda_{1,\sdimn}^{\vnorm{X},\vnorm{Y}}
\leq
\underset{X\sim_{\rho}Y}{\E}\Big(\langle \E f_{\vnorm{X}}, \E f_{\vnorm{Y}}\rangle
-\vnorm{\E f_{\vnorm{X}}}^{2}\lambda_{1,\sdimn}^{\vnorm{X},\vnorm{Y}}\Big).
\end{equation}
It remains to show that the right side is nonpositive.  Since $\E_{\gamma}f\colonequals\int_{\R^{n}}f(x)\gamma_{n}(x)\,\d x=0$, we have by Lemma \ref{lemma28} that
\begin{equation}\label{four1}
\begin{aligned}
\underset{X\sim_{\rho}Y}{\E}\langle \E f_{\vnorm{X}}, \E f_{\vnorm{Y}}\rangle
&=\underset{X\sim_{\rho}Y}{\E}\langle \E f_{\vnorm{X}}-\E_{\gamma}f, \E f_{\vnorm{Y}}-\E_{\gamma}f\rangle\\
&\leq9.4\rho\int_{\R^{n}}\phi(x)\vnorm{\E f_{\vnorm{x}}}^{2}\gamma_{n}(x)\, d x.
\end{aligned}
\end{equation}
Here $\phi(x)\colonequals 1 - \frac{1}{\rho\vnorm{x}}+\frac{2}{\exp^{2\rho\vnorm{x}} -1}\stackrel{\eqref{one9}}{=}\lambda_{1,n}^{\vnorm{x},1}$.  Meanwhile, the last term in \eqref{four0} satisfies
\begin{equation}\label{four2}
\begin{aligned}
&\underset{X\sim_{\rho}Y}{\E}\vnorm{\E f_{\vnorm{X}}}^{2}\lambda_{1,\sdimn}^{\vnorm{X},\vnorm{Y}}
\stackrel{\eqref{oudef}}{=}\int_{\R^{3}}\vnorm{\E f_{\vnorm{x}}}^{2}\Big(\int_{\R^{3}}\lambda_{1,\sdimn}^{\vnormf{\rho x+y\sqrt{1-\rho^{2}}},\vnorm{x}}\gamma_{n}(y)\,\d y\Big)\gamma_{n}(x)\,\d x\\
&\qquad=\int_{\R^{3}}\vnorm{\E f_{\vnorm{x}}}^{2}\Big(\sqrt{\frac{2}{\pi}}\frac{1}{4\pi}\int_{u\in S^{2}}\int_{r=0}^{\infty}r^{2}\lambda_{1,\sdimn}^{\vnormf{\rho x+ru\sqrt{1-\rho^{2}}},\vnorm{x}}e^{-\vnorm{r}^{2}/2}\,\d r \d u\Big)\gamma_{n}(x)\,\d x\\
&\qquad\stackrel{\eqref{one9}}{\geq}\int_{\R^{3}}\vnorm{\E f_{\vnorm{x}}}^{2}\Big(\sqrt{\frac{2}{\pi}}\frac{1}{4\pi}\int_{u\in S^{2}}\int_{r=0}^{\infty}r^{2}\lambda_{1,\sdimn}^{\vnormf{ru\sqrt{1-\rho^{2}}},\vnorm{x}}e^{-\vnorm{r}^{2}/2}\,\d r \d u\Big)\gamma_{n}(x)\,\d x\\
&\qquad=\int_{\R^{3}}\vnorm{\E f_{\vnorm{x}}}^{2}\Big(\sqrt{\frac{2}{\pi}}\int_{r=0}^{\infty}\lambda_{1,\sdimn}^{r\sqrt{1-\rho^{2}},\vnorm{x}}e^{-\vnorm{r}^{2}/2}\,\d r\Big)\gamma_{n}(x)\,\d x.
\end{aligned}
\end{equation}
The inequality used that $\lambda_{1,2}^{a,b}$ is an increasing function of $a>0$ by \eqref{one9}, so the average over $r,u$ is smallest when $\rho x=0$.

From \eqref{two9} when $n=3$ and $a=\rho rs/(1-\rho^{2})$ with \eqref{lameq}, if $s>0$ and $0<\rho<1/5$, then
\begin{flalign*}
\sqrt{\frac{2}{\pi}}\int_{r=0}^{\infty}r^{2}\lambda_{1,\sdimn}^{r\sqrt{1-\rho^{2}},s}e^{-\vnorm{r}^{2}/2}\,\d r
&\stackrel{\eqref{two9}\wedge\eqref{lameq}}{\geq} \sqrt{\frac{2}{\pi}}e^{\frac{9(1-\rho^{2})}{8s^{2}\rho^{2}}}\int_{\frac{3\sqrt{1-\rho^{2}}}{2\rho s}}^{\infty}e^{-t^{2}/2}\d t\\
&\geq(.98)\Big(1 - \frac{1}{\rho s}+\frac{2}{e^{2\rho s} -1}\Big).
\end{flalign*}

%

The last inequality can be verified e.g. with Matlab
\begin{verbatim}
rho=.03;
r=linspace(0,1000,1000);
phi=.98*(1 -  1./(rho * r ) +2./(exp(2* rho *r ) -1));
y=exp(  9*(1-rho^2) ./ (8* r.^2 * rho^2))  ...
 .* erfc( 3*sqrt(1-rho^2) ./ (2*sqrt(2) * rho *r));
plot(r, phi, r, y);
legend('lambda lower bd','integral quant (larger)');
if sum(y-phi<0)==0, fprintf('Verified\r'), end
\end{verbatim}

Substituting this into \eqref{four2}, we get

\begin{equation}\label{four3}
\underset{X\sim_{\rho}Y}{\E}\vnorm{\E f_{\vnorm{X}}}^{2}\lambda_{1,\sdimn}^{\vnorm{X},\vnorm{Y}}
\geq.98\int_{\R^{3}}\phi(x)\vnorm{\E f_{\vnorm{x}}}^{2}\gamma_{n}(x)\,\d x.
\end{equation}

Combining \eqref{four0}, \eqref{four3} and \eqref{four1}, we have
\begin{equation}\label{four4}
\underset{X\sim_{\rho}Y}{\E}\langle f(X),f(Y)\rangle
-\underset{X\sim_{\rho}Y}{\E}\lambda_{1,\sdimn}^{\vnorm{X},\vnorm{Y}}
\leq(9.4\rho - .98)\int_{\R^{n}}\phi(x)\vnorm{\E f_{\vnorm{x}}}^{2}\gamma_{n}(x)\, d x.
\end{equation}
If $\rho<.104$, the right side is nonpositive, and it is equal to zero only when $\E f_{\vnorm{x}}=0$ for a.e. $x\in\R^{3}$.  That is,
$$\underset{X\sim_{\rho}Y}{\E}\langle f(X),f(Y)\rangle
\stackrel{\eqref{four4}}{\leq}\underset{X\sim_{\rho}Y}{\E}\lambda_{1,\sdimn}^{\vnorm{X},\vnorm{Y}}
\stackrel{\eqref{foptdef}\wedge\eqref{nine7}}{=}\underset{X\sim_{\rho}Y}{\E}\langle f_{\rm opt}(X),f_{\rm opt}(Y)\rangle.$$
with equality only when $\E f_{\vnorm{x}}=0$ for a.e. $x\in\R^{3}$.  Finally, if $\E f_{\vnorm{x}}=0$ for a.e. $x\in\R^{3}$, then Lemma \ref{lemma2} implies that we must have $f=f_{\rm opt}(M\cdot)$ for some real $3\times 3$ orthogonal matrix $M$.
\end{proof}

\section{Proof of Main Theorem: Bilinear Case}\label{seclast}

\begin{proof}[Proof of Theorem \ref{thm1} when $\rho<0$]
In order to prove Theorem \ref{thm1} for negative $\rho$, we consider instead $\rho>0$ but with a bilinear variant of the noise stability over functions $f,g\colon\R^{n}\to S^{2}$ under the constraint that $\E_{\gamma}f=\E_{\gamma}g=0$.

So, within the proof below, we have $\rho>0$, and we will show that
\begin{equation}\label{six1}
\underset{X\sim_{\rho}Y}{\E}\langle f(X),g(Y)\rangle\geq -\underset{X\sim_{\rho}Y}{\E}\langle f_{\rm opt}(X),f_{\rm opt}(Y)\rangle.
\end{equation}
Dimension Reduction (Theorem \ref{thm9} below) implies that we may assume $n=k=3$ in order to prove \eqref{six1}.  Equation \eqref{six1} proves Theorem \ref{thm1} for negative correlations, since
\begin{flalign*}
\underset{X\sim_{(-\rho)}Y}{\E}\langle f(X),f(Y)\rangle
&\stackrel{\eqref{gdef}}=\underset{X\sim_{\rho}Y}{\E}\langle f(X),f(-(Y))\rangle
\stackrel{\eqref{six1}}{\geq}-\underset{X\sim_{\rho}Y}{\E}\langle f_{\rm opt}(X),f_{\rm opt}(Y)\rangle\\
&\stackrel{\eqref{gdef}}=\underset{X\sim_{(-\rho)}Y}{\E}\langle f_{\rm opt}(X),f_{\rm opt}(-Y)\rangle
\stackrel{\eqref{foptdef}}{=}\underset{X\sim_{(-\rho)}Y}{\E}\langle f_{\rm opt}(X),-f_{\rm opt}(Y)\rangle.
\end{flalign*}

So, it remains to prove \eqref{six1}.  From Lemma \ref{lemma7}
\begin{equation}\label{five0}
\underset{X\sim_{\rho}Y}{\E}\Big(\langle f(X),g(Y)\rangle
+\lambda_{1,\sdimn}^{\vnorm{X},\vnorm{Y}}\Big)
\geq
\underset{X\sim_{\rho}Y}{\E}\Big(\langle \E f_{\vnorm{X}}, \E g_{\vnorm{Y}}\rangle
+\frac{1}{2}(\vnorm{\E f_{\vnorm{X}}}^{2}+\vnorm{\E g_{\vnorm{X}}}^{2})\lambda_{1,\sdimn}^{\vnorm{X},\vnorm{Y}}\Big).
\end{equation}
It remains to show that the right side is nonnegative.  Since $\E_{\gamma}f=\E_{\gamma}g$ by assumption,
\begin{flalign*}
\underset{X\sim_{\rho}Y}{\E}\langle \E f_{\vnorm{X}}, \E g_{\vnorm{Y}}\rangle
&=\vnorm{\E_{\gamma}f}^{2}+\underset{X\sim_{\rho}Y}{\E}\langle \E f_{\vnorm{X}}- \E_{\gamma}f, \E g_{\vnorm{Y}}-\E_{\gamma}g\rangle\\
&\geq \underset{X\sim_{\rho}Y}{\E}\langle \E f_{\vnorm{X}}- \E_{\gamma}f, \E g_{\vnorm{Y}}-\E_{\gamma}g\rangle.
\end{flalign*}
So, by Lemma \ref{lemma29},
\begin{equation}\label{five1}
\begin{aligned}
\underset{X\sim_{\rho}Y}{\E}\langle \E f_{\vnorm{X}}, \E g_{\vnorm{Y}}\rangle
&\geq\underset{X\sim_{\rho}Y}{\E}\langle \E f_{\vnorm{X}}-\E_{\gamma}f, \E g_{\vnorm{Y}}-\E_{\gamma}g\rangle\\
&\geq-9.4\rho\int_{\R^{n}}\phi(x)\frac{1}{2}(\vnorm{\E f_{\vnorm{x}}}^{2}+\vnorm{\E g_{\vnorm{x}}}^{2})\gamma_{n}(x)\, d x.
\end{aligned}
\end{equation}
Here $\phi(x)\colonequals 1 - \frac{1}{\rho\vnorm{x}}+\frac{2}{\exp^{2\rho\vnorm{x}} -1}\stackrel{\eqref{one9}}{=}\lambda_{1,n}^{\vnorm{x},1}$.  Meanwhile, \eqref{four3} says that

\begin{equation}\label{five3}
\underset{X\sim_{\rho}Y}{\E}\vnorm{\E f_{\vnorm{X}}}^{2}\lambda_{1,\sdimn}^{\vnorm{X},\vnorm{Y}}
\geq.98\int_{\R^{3}}\phi(x)\vnorm{\E f_{\vnorm{x}}}^{2}\gamma_{n}(x)\,\d x.
\end{equation}
\begin{equation}\label{five3p2}
\underset{X\sim_{\rho}Y}{\E}\vnorm{\E g_{\vnorm{X}}}^{2}\lambda_{1,\sdimn}^{\vnorm{X},\vnorm{Y}}
\geq.98\int_{\R^{3}}\phi(x)\vnorm{\E g_{\vnorm{x}}}^{2}\gamma_{n}(x)\,\d x.
\end{equation}
Adding these together, we get
\begin{equation}\label{five3p3}
\underset{X\sim_{\rho}Y}{\E}\frac{1}{2}(\vnorm{\E f_{\vnorm{X}}}^{2}+\vnorm{\E g_{\vnorm{X}}}^{2})\lambda_{1,\sdimn}^{\vnorm{X},\vnorm{Y}}
\geq.98\int_{\R^{3}}\phi(x)\frac{1}{2}(\vnorm{\E f_{\vnorm{x}}}^{2}+\vnorm{\E g_{\vnorm{x}}}^{2})\gamma_{n}(x)\,\d x.
\end{equation}

Combining \eqref{five0}, \eqref{five3p3} and \eqref{five1}, we have
\begin{equation}\label{five4}
\begin{aligned}
&\underset{X\sim_{\rho}Y}{\E}\langle f(X),g(Y)\rangle
+\underset{X\sim_{\rho}Y}{\E}\lambda_{1,\sdimn}^{\vnorm{X},\vnorm{Y}}\\
&\qquad\qquad\qquad\qquad\geq(.98-9.4\rho)\int_{\R^{n}}\phi(x)\frac{1}{2}(\vnorm{\E f_{\vnorm{x}}(x)}^{2}+\vnorm{\E g_{\vnorm{x}}(x)}^{2})\gamma_{n}(x)\, d x.
\end{aligned}
\end{equation}
If $\rho<.104$, the right side is nonnegative, and it is equal to zero only when $\E f_{\vnorm{x}}=\E g_{\vnorm{x}}=0$ for a.e. $x\in\R^{3}$.  That is,
$$\underset{X\sim_{\rho}Y}{\E}\langle f(X),g(Y)\rangle
\stackrel{\eqref{five4}}{\geq}-\underset{X\sim_{\rho}Y}{\E}\lambda_{1,\sdimn}^{\vnorm{X},\vnorm{Y}}
\stackrel{\eqref{foptdef}\wedge\eqref{nine7}}{=}-\underset{X\sim_{\rho}Y}{\E}\langle f_{\rm opt}(X),f_{\rm opt}(Y)\rangle.$$
with equality only when $\E f_{\vnorm{x}}=\E g_{\vnorm{x}}=0$ for a.e. $x\in\R^{3}$.

If $\E f_{\vnorm{x}}=\E g_{\vnorm{x}}=0$ for all $x\in\R^{3}$, then Lemma \ref{lemma6} implies that we must have $f=-g=f_{\rm opt}(M\cdot)$ almost surely, for some real $3\times 3$ orthogonal matrix $M$.  So, \eqref{six1} is proven, and the negative correlation case of Theorem \ref{thm1} follows.
\end{proof}

\section{Proof of Unique Games Hardness}

\begin{proof}[Proof of Theorem \ref{thm7}]
From \cite[Theorem 11.3]{hwang21}: assuming the Unique Games Conjecture, for any $-1<\rho<0$ and for any $\epsilon>0$, it is NP-hard to approximate the product state of Quantum MAX-CUT within a multiplicative factor of $\alpha_{\rho,\rm BOV}+\epsilon$, where
$$
\alpha_{\rho,\rm BOV}
\colonequals
\lim_{n\to\infty}\sup_{f\colon\R^{n}\to S^{2}}\frac{1-\int_{\R^{n}}\langle f(x), T_{\rho}f(x)\rangle\,\gamma_{n}(x)\,\d x}{1-\rho}.
$$
Theorem \ref{thm1} says, if $-.104<\rho<0$, then
\begin{equation}\label{labeq}
\alpha_{\rho,\rm BOV}=\frac{1-\int_{\R^{3}}\langle f_{\rm opt}(x), T_{\rho}f_{\rm opt}(x)\rangle\,\gamma_{3}(x)\,\d x}{1-\rho}.
\end{equation}
An explicit formula for the noise stability of $f_{\rm opt}$ from \cite[Proposition 7.17]{hwang21} implies
$$
\alpha_{\rho,\rm BOV}
\stackrel{\eqref{foptdef}\wedge\eqref{labeq}}{=}\frac{1-F^{*}(3,\rho)}{1-\rho}
\stackrel{\eqref{fstdef}}{=}\frac{1-\frac{2}{3}\Big(\frac{\Gamma((3+1)/2)}{\Gamma(3/2)}\Big)^{2}\rho\cdot\,\,_{2}F_{1}(1/2,1/2,3/2 +1, \rho^{2})}{1-\rho}.
$$
Since NP hardness then applies from \cite[Theorem 11.3]{hwang21} for any $-.104<\rho<0$, we then have NP-hardness for the infimum of $\alpha_{\rho,\rm BOV}$ over all $-.104<\rho<0$.  That is, it is NP-hard to approximate the product state of Quantum MAX-CUT within a multiplicative factor of
\begin{flalign*}
\inf_{-.104<\rho<0}\alpha_{\rho,\rm BOV}
&=\alpha_{(-.104),\rm BOV}
=\frac{1-\frac{2}{3}\Big(\frac{\Gamma((3+1)/2)}{\Gamma(3/2)}\Big)^{2}(-.104)\cdot\,\,_{2}F_{1}(1/2,1/2,3/2 +1, (.104)^{2})}{1-(-.104)}\\
&=.98584579\ldots.
\end{flalign*}

The above function of $\rho$ is monotone, with minimum at the endpoint $\rho=-.104$.
\end{proof}

\section{Appendix: Dimension Reduction}

\begin{theorem}[\embolden{Bilinear Dimension Reduction}]\label{thm9}
Let $\rho>0$.  Define
$$s_{n}\colonequals \inf_{\substack{f,g\colon\R^{n}\to S^{k-1}\\\E_{\gamma}f=\E_{\gamma}g=0 }}\quad
\underset{X\sim_{\rho}Y}{\E}\langle f(X),g(Y)\rangle.$$
Then
$$s_{n}=s_{k},\qquad\forall\, n\geq k.$$
\end{theorem}
The inequality $s_{n}\leq s_{k}$ easily holds for all $n\geq k$, so the content of Theorem \ref{thm9} is that $s_{n}\geq s_{k}$ for all $n\geq k$.

Theorem \ref{thm9} is a fairly straightforward adaptation of \cite[Theorem 6.11]{hwang21} to the bilinear setting (itself an adaptation of \cite{heilman20d} to the sphere-valued setting).  In order to emphasize the relation between the argument below and that of \cite[Theorem 6.11]{hwang21}, we will match the notation from \cite{hwang21}, where appropriate.

We say a vector field $W\colon\R^{n}\to\R^{k}$ is a \textbf{tame vector field} if $W$ is bounded and $C^{\infty}$, and all of the partial derivatives of $W$ of any order exists and are bounded.

For any $f\colon\R^{n}\to\R^{k}$, for any $t\in\R$, denote
\begin{equation}\label{vdef}
\mathcal{V}_{t,W}f(x)\colonequals\frac{f(x)+tW(x)}{\max(1,\vnorm{f(x)+tW(x)})},\qquad\forall\,x\in\R^{n}.
\end{equation}
Denote also
$$B^{k}\colonequals\{x\in\R^{k}\colon\vnorm{x}\leq 1\}.$$
\begin{definition}
Let $f,g\colon\R^{n}\to B^{k}$ satisfy $\E_{\gamma}f=\E_{\gamma}g=0$.  Fix $0<\rho<1$.  We say that $f,g$ are \textbf{optimally stable} with correlation $\rho$ if, for all $h,k\colon\R^{n}\to B^{k}$ with $\E_{\gamma}h=\E_{\gamma}k=0$,
$$\int_{\R^{k}}\langle f(x), T_{\rho}g(x)\rangle \gamma_{k}(x)\, \d x\leq \int_{\R^{k}}\langle h(x), T_{\rho}k(x)\rangle \gamma_{k}(x)\, \d x.$$
\end{definition}

\begin{lemma}\label{lemma0}
If $f,g$ are optimally stable with correlation $\rho$, then $g=-f$.
\end{lemma}
\begin{proof}
Since $f=g=0$ has $\E_{\gamma}\langle f,T_{\rho}g\rangle=0$, if $f,g$ are optimally stable, then $\E_{\gamma}\langle f,T_{\rho}g\rangle\leq0$.  So, the Cauchy-Schwarz inequality implies that
\begin{flalign*}
\E_{\gamma}\langle f,T_{\rho}g\rangle
&=-\abs{\E_{\gamma}\langle f,T_{\rho}g\rangle}\geq -\abs{\E_{\gamma}\langle f,T_{\rho}f\rangle}^{1/2}\abs{\E_{\gamma}\langle g,T_{\rho}g\rangle}^{1/2}\\
&\geq-\max\Big(\abs{\E_{\gamma}\langle f,T_{\rho}f\rangle},\abs{\E_{\gamma}\langle g,T_{\rho}g\rangle}\Big).
\end{flalign*}
All of these inequalities become equalities when $g=-f$, since $\E_{\gamma}\langle f,T_{\rho}f\rangle\geq0$, so
$$
\E_{\gamma}\langle f,T_{\rho}(-f)\rangle
=-\E_{\gamma}\langle f,T_{\rho}f\rangle
=-\abs{\E_{\gamma}\langle f,T_{\rho}f\rangle}.
$$
\end{proof}

\begin{lemma}[Lemma 6.8, {\cite{hwang21}}]\label{lemma30}
Let $f\colon\R^{n}\to\R^{k}$ be measurable.  Then there exists vectors fields $W_{1},\ldots,W_{k}\colon\R^{n}\to\R^{k}$ such that
$$\mathrm{span}\Big\{\frac{\d}{\d t}\Big|_{t=0}\E_{\gamma} [\mathcal{V}_{t,W}f]\colon 1\leq i\leq k\Big\}=\R^{k}.$$
\end{lemma}

\begin{lemma}[Lemma 6.9, {\cite{hwang21}}]\label{lemma31}
Let $f,g\colon\R^{n}\to\R^{k}$ be optimally stable.  Then, for any bounded measurable vector fields $W,Z\colon\R^{n}\to\R^{k}$ such that $\frac{\d}{\d t}|_{t=0}\E_{\gamma} \mathcal{V}_{t,W}f=\frac{\d}{\d t}|_{t=0}\E_{\gamma} \mathcal{V}_{t,Z}g=0$, we have
$$\frac{\d}{\d t}\Big|_{t=0}\underset{X\sim_{\rho}Y}{\E}\langle\mathcal{V}_{t,W}f(X),\, \mathcal{V}_{t,Z}g(Y)\rangle=0.$$
\end{lemma}

\begin{proof}
$$f_{\alpha,\beta}\colonequals\frac{f(x)+\beta W(x)+\sum_{i=1}^{k}\alpha_{i}W_{i}(x)  }{\max\Big(1,\vnormf{f(x)+\beta W(x)+\sum_{i=1}^{k}\alpha_{i}W_{i}(x)}\Big)}.$$
$$g_{\alpha,\beta}\colonequals\frac{g(x)+\beta Z(x)+\sum_{i=1}^{k}\alpha_{i}W_{i}(x)  }{\max\Big(1,\vnormf{g(x)+\beta Z(x)+\sum_{i=1}^{k}\alpha_{i}W_{i}(x)}\Big)}.$$

Define $L\colon\R^{2k+1}\to\R^{2k}$ by
$$L(\alpha,\theta,\beta)\colonequals \left(\E_{\gamma}f_{\alpha,\beta},\,\, \E_{\gamma}g_{\theta\beta}\right),\qquad\forall\,\alpha,\theta\in\R^{k},\,\forall\,\beta\in\R.$$
Then by assumption, we have
\begin{equation}\label{three1}
\frac{\partial L}{\partial\beta}(0,0)\stackrel{\eqref{vdef}}{=}\left(\frac{\d}{\d t}\Big|_{t=0}\E_{\gamma}\mathcal{V}_{t,W}f,\,\,\frac{\d}{\d t}\Big|_{t=0}\E_{\gamma}\mathcal{V}_{t,Z}g\right)
=0.
\end{equation}
And by definition of $L$, we have, for all $1\leq i\leq k$,
\begin{equation}\label{three2}
\begin{aligned}
\frac{\partial L}{\partial\alpha_{i}}(0,0)&=\left(\frac{\d}{\d t}\Big|_{t=0}\E_{\gamma}\mathcal{V}_{t,W_{i}}f,\,\, 0 \right).\\
\frac{\partial L}{\partial\theta_{i}}(0,0)&=\left(  0,\,\, \frac{\d}{\d t}\Big|_{t=0}\E_{\gamma}\mathcal{V}_{t,W_{i}}g \right).\\
\end{aligned}
\end{equation}
From Lemma \ref{lemma30}, we conclude that the matrix of partial derivatives $DL$ of $L$ is a $(2k+1)\times 2k$ matrix of rank $2k$.  So, by the Implicit Function Theorem, there exists $\epsilon>0$ and a differentiable curve $\eta\colon(-\epsilon,\epsilon)\to\R^{2k+1}$ with $\eta(0)=0$, $\eta'(0)\neq0$ and $L(\eta(t))=0$ for all $t\in(-\epsilon,\epsilon)$.  The last property implies, by the Chain Rule, that
\begin{equation}\label{three3}
0=\frac{\d}{\d t}\Big|_{t=0}L(\eta(t))
=\sum_{i=1}^{k}\frac{\partial L}{\partial\alpha_{i}}(0,0)\eta_{i}'(0)
+\sum_{i=1}^{k}\frac{\partial L}{\partial\theta_{i}}(0,0)\eta_{k+i}'(0)\qquad+\frac{\partial L}{\partial \beta}(0,0)\eta_{2k+1}'(0).\\
\end{equation}
The last term is zero by \eqref{three1}.  Lemma \ref{lemma30} and \eqref{three2} imply that set $\Big\{\frac{\partial L}{\partial\alpha_{i}}(0,0),\frac{\partial L}{\partial\theta_{i}}(0,0)\Big\}_{i=1}^{k}$ consists of $2k$ linearly independent vectors.  We conclude from \eqref{three3} that $\eta_{i}'(0)=\eta_{i+k}'(0)=0$ for all $1\leq i\leq k$.  Since $\eta'(0)\neq0$, we conclude that $\eta_{2k+1}'(0)\neq0$.  Let $J(t)\colonequals \underset{X\sim_{\rho}Y}{\E}\langle f_{(\eta_{i}(t))_{i=1}^{k},\eta_{2k+1}(t)}(X),\, g_{(\eta_{i+k}(t))_{i=1}^{k},\eta_{2k+1}(t)}(Y)\rangle$ for all $t\in(-\epsilon,\epsilon)$.  Since $L(\eta(t))=0$ for all $t\in(-\epsilon,\epsilon)$, $(f,g)$ satisfy $\E_{\gamma}f=\E_{\gamma}g=0$ for all $t\in(-\epsilon,\epsilon)$.  Since $(f,g)$ are optimally stable, we therefore have $(d/dt)J(t)=0$.  Using the previous facts and the chain rule,
\begin{flalign*}
0&=\frac{\d}{\d t}\Big|_{t=0}J(t)
=\sum_{i=1}^{k}\eta_{i}'(t)\frac{\partial }{\partial\alpha_{i}}\underset{X\sim_{\rho}Y}{\E}\langle f_{\alpha,\beta}(X),\, g_{\theta,\beta}(Y)\rangle
+\sum_{i=1}^{k}\eta_{i+k}'(t)\frac{\partial }{\partial\theta_{i}}\underset{X\sim_{\rho}Y}{\E}\langle f_{\alpha,\beta}(X),\, g_{\theta,\beta}(Y)\rangle\\
&\qquad+\eta_{2k+1}'(0)\frac{\d}{\d t}\Big|_{t=0}\underset{X\sim_{\rho}Y}{\E}\langle f_{\alpha,\eta_{2k+1}(t)}(X),\, g_{\theta,\eta_{2k+1}(t)}(Y)\rangle\\
&=\eta_{2k+1}'(0)\frac{\d}{\d t}\Big|_{t=0}\underset{X\sim_{\rho}Y}{\E}\langle f_{\alpha,\eta_{2k+1}(t)}(X),\, g_{\theta,\eta_{2k+1}(t)}(Y)\rangle\\
&\stackrel{\eqref{vdef}}{=}\eta_{2k+1}'(0)\frac{\d}{\d t}\Big|_{t=0}\underset{X\sim_{\rho}Y}{\E}\langle\mathcal{V}_{t,W}f(X),\, \mathcal{V}_{t,Z}g(Y)\rangle.
\end{flalign*}
Since $\eta_{2k+1}'(0)\neq0$, the last term is zero, i.e. the proof is concluded.
\end{proof}

\begin{lemma}[{\cite[Lemma 6.10]{hwang21}}]\label{lemma31Z}
Let $f,g$ be optimally stable.  $\exists$ $\lambda\in\R^{2k}$ such that
$$
\frac{\d}{\d t}\Big|_{t=0}\underset{X\sim_{\rho}Y}{\E}\langle\mathcal{V}_{t,W}f(X),\, \mathcal{V}_{t,Z}g(Y)\rangle
=\Big\langle \Big(\frac{\d}{\d t}\Big|_{t=0}\E_{\gamma}\mathcal{V}_{t,W}f ,\, \frac{\d}{\d t}\Big|_{t=0}\E_{\gamma}\mathcal{V}_{t,Z}g\Big),\lambda\Big\rangle.
$$
\end{lemma}
\begin{proof}
Let $W,Z\colon\R^{n}\to\R^{n}$, and define $\phi_{1},\phi_{2},\psi$ by
$$
\phi_{1}(W)\colonequals\frac{\d}{\d t}\Big|_{t=0}\E_{\gamma}\mathcal{V}_{t,W}f.
$$
$$
\phi_{2}(Z)\colonequals\frac{\d}{\d t}\Big|_{t=0}\E_{\gamma}\mathcal{V}_{t,Z}g.
$$
$$\psi(W,Z)\colonequals\frac{\d}{\d t}\Big|_{t=0}\underset{X\sim_{\rho}Y}{\E}\langle\mathcal{V}_{t,W}f(X),\, \mathcal{V}_{t,Z}g(Y)\rangle.$$
Note that $\phi_{1},\phi_{2},\psi$ are linear functions of $(W,Z)$.  Let $\mathcal{X}$ be a finite-dimensional subspace of bounded measurable vectors fields such that $\{(\phi_{1}(W),\phi_{2}(Z))\colon W,Z\in\mathcal{X}\}$ spans $\R^{2k}$.  Define $L\colon \mathcal{X}\to \R^{2k+1}$ by

$$L(W,Z)\colonequals\Big(\phi_{1}(W),\phi_{2}(Z),\psi(W,Z)\Big).$$
From Lemma \ref{lemma30}, $(0,\ldots,0,1)$ is not in the range of $L$.  So, there exists $\lambda=\lambda(\mathcal{X})\in\R^{2k}$ such that $(-\lambda,1)$ is orthogonal to the range of $L$ (with respect to the standard inner product in $\R^{2k+1}$.)   That is, $\langle L(W,Z), (-\lambda,1)\rangle=0$ for all $W,Z\in\mathcal{X}$.  That is, $\psi(W,Z)=\langle(\phi_{1}(W),\phi_{2}(Z)),\lambda\rangle$.  As in \cite[Lemma 6.10]{hwang21}, $\lambda$ does not depend on $\mathcal{X}$.
\end{proof}

For any $t\in\R$, denote $t_{+}\colonequals\max(t,0)$.
\begin{lemma}[Lemma 6.7, {\cite{hwang21}}]\label{lemma32}
Let $f\colon\R^{n}\to B^{k}$ be measurable.  Let $W\colon\R^{n}\to\R^{k}$ be a bounded measurable vector field.  For any $x\in\R^{n}$,
$$
\mathcal{V}_{t,W}f(x)= f(x)+tW(x)-t\langle W(x),f(x)\rangle_{+}f(x)+O(t^{2}),
$$
where the $O(t^{2})$ term is uniform in $x$.  Also,
$$
\frac{\d}{\d t}|_{t=0}\E_{\gamma}\mathcal{V}_{t,W}f
=\E_{\gamma}[W - \langle f,W\rangle_{+}f1_{\{\vnorm{f}=1\}}].
$$
\end{lemma}

We write $\lambda$ from Lemma \ref{lemma31} as $\lambda=(\lambda^{(1)},\lambda^{(2)})\in\R^{2k}$ so that $\lambda^{(i)}\in\R^{k}$ for $i=1,2$.

\begin{lemma}[\embolden{First Variation}, Lemma 6.11, {\cite{hwang21}}]\label{lemma33}
Let $f,g$ be optimally stable.  Then
$$\vnormf{T_{\rho}f-\lambda^{(1)}}g=T_{\rho}f-\lambda^{(1)},\qquad\vnormf{T_{\rho}g-\lambda^{(2)}}f=T_{\rho}g-\lambda^{(2)}.\qquad\mbox{a.s.}$$
\end{lemma}
\begin{proof}
Combining Lemmas \ref{lemma32} and \ref{lemma31Z}, Definition \ref{oudef}, Lemma \ref{lemma32} again, and Definition \ref{gausdef},
\begin{flalign*}
&\Big\langle \Big(  \E_{\gamma}[W - \langle f,W\rangle_{+}f1_{\{\vnorm{f}=1\}}],\, \E_{\gamma}[Z - \langle g,Z\rangle_{+}g1_{\{\vnorm{g}=1\}}]\Big),\lambda\Big\rangle\\
&=\Big\langle \Big(\frac{\d}{\d t}\Big|_{t=0}\E_{\gamma}\mathcal{V}_{t,W}f ,\, \frac{\d}{\d t}\Big|_{t=0}\E_{\gamma}\mathcal{V}_{t,Z}g\Big),\lambda\Big\rangle\\
&=\frac{\d}{\d t}\Big|_{t=0}\underset{X\sim_{\rho}Y}{\E}\langle\mathcal{V}_{t,W}f(X),\, \mathcal{V}_{t,Z}g(Y)\rangle\\
&=\underset{X\sim_{\rho}Y}{\E}\langle f(X),\, Z(Y)-\langle Z(Y),g(Y)\rangle_{+}g(Y)\rangle
+\underset{X\sim_{\rho}Y}{\E}\langle W(X)-\langle W(X),f(X)\rangle_{+}f(X),\, g(Y)\rangle\\
&=\E_{\gamma}\langle T_{\rho}f,\, Z-\langle Z,g\rangle_{+}g\rangle
+\E_{\gamma}\langle W-\langle W,f\rangle_{+}f,\, T_{\rho}g\rangle.
\end{flalign*} 
The above equality can be rearranged to be
$$
\E_{\gamma}\langle T_{\rho}f -\lambda^{(1)},\, Z-\langle Z,g\rangle_{+}g\rangle
+\E_{\gamma}\langle  T_{\rho}g - \lambda^{(2)},\,W-\langle W,f\rangle_{+}f\rangle
=0.
$$
The proof then concludes as in \cite[Lemma 6.11]{hwang21}.
\end{proof}

Let $W\colon\R^{n}\to\R^{n}$ be a tame vector field.  Define
$$F_{0}(x)\colonequals x,\quad \frac{\d}{\d t}F_{t}(x)=W(F_{t}(x)),\qquad\forall\,t\in\R,\,\forall\,x\in\R^{n}.$$
Denote $F_{t}=F_{t,W}$, where appropriate.
\begin{equation}\label{sdef}
\mathcal{S}_{t,W}f(x)\colonequals f(F_{t}^{-1}(x)),\qquad\forall\,t\in\R,\,\forall\,x\in\R^{n}.
\end{equation}
$$\mathrm{div}_{\gamma}W(x)\colonequals\mathrm{div}W(x)-\langle W(x),x\rangle.$$

\begin{lemma}[\embolden{First Variation}, Lemma 6.13, {\cite{hwang21}}]\label{lemma34}
Let $f,g\colon\R^{n}\to\R^{k}$ and let $W,Z\colon \R^{n}\to\R^{n}$ be tame vector fields.  Then
\begin{flalign*}
\frac{\d}{\d t}\Big|_{t=0}\underset{X\sim_{\rho}Y}{\E}\langle\mathcal{S}_{t,W}f(X),\, \mathcal{S}_{t,Z}g(Y)\rangle
&=\E_{\gamma}[\langle f,T_{\rho}f\rangle\mathrm{div}_{\gamma}W + \langle f, D_{W}T_{\rho}f\rangle]\\
&\quad+\E_{\gamma}[\langle g,T_{\rho}g\rangle\mathrm{div}_{\gamma}Z + \langle g, D_{Z}T_{\rho}g\rangle].
\end{flalign*}
\end{lemma}
Let $\lambda\in\R^{2k}$ from Lemma \ref{lemma31Z}.  Define
\begin{equation}\label{qdef}
Q(W,Z)\colonequals\frac{\d^{2}}{\d t^{2}}\Big|_{t=0}\underset{X\sim_{\rho}Y}{\E}\langle\mathcal{S}_{t,W}f(X),\, \mathcal{S}_{t,Z}g(Y)\rangle
-\Big\langle\lambda,\,\Big(\frac{\d^{2}}{\d t^{2}}\Big|_{t=0}\E_{\gamma} \mathcal{S}_{t,W}f ,\, \frac{\d^{2}}{\d t^{2}}\Big|_{t=0}\E_{\gamma} \mathcal{S}_{t,Z}g \Big) \Big\rangle.
\end{equation}

\begin{lemma}[\embolden{Second Variation}, Lemma 6.15, {\cite{hwang21}}]\label{lemma35}
Let $W,Z\colon\R^{n}\to\R^{n}$ be tame vector fields such that
$$\frac{\d}{\d t}\Big|_{t=0}\E_{\gamma}\mathcal{S}_{t,W}f=\frac{\d}{\d t}\Big|_{t=0}\E_{\gamma}\mathcal{S}_{t,Z}g=0.$$
Let $f,g$ be optimally stable.  Then
$$Q(W,Z)\geq0.$$
\end{lemma}
\begin{proof}
Similar to Lemma \ref{lemma31}, for any $\alpha\in\R^{k}$ and $\beta\in\R$, define

$$f_{\alpha,\beta}\colonequals\frac{f(F_{\beta,W}(x))+\sum_{i=1}^{k}\alpha_{i}W_{i}(x)  }{\max\Big(1,\vnormf{f(F_{\beta,W}(x))+\sum_{i=1}^{k}\alpha_{i}W_{i}(x)}\Big)}.$$
$$g_{\alpha,\beta}\colonequals\frac{g(F_{\beta,Z}(x))+\sum_{i=1}^{k}\alpha_{i}W_{i}(x)  }{\max\Big(1,\vnormf{g(F_{\beta,Z}(x))+\sum_{i=1}^{k}\alpha_{i}W_{i}(x)}\Big)}.$$

Define $L\colon\R^{2k+1}\to\R^{2k}$ by
$$L(\alpha,\theta,\beta)\colonequals \left(\E_{\gamma}f_{\alpha,\beta},\,\, \E_{\gamma}g_{\theta\beta}\right),\qquad\forall\,\alpha,\theta\in\R^{k},\,\forall\,\beta\in\R.$$
Then by assumption, we have
\begin{equation}\label{three1z}
\frac{\partial L}{\partial\beta}(0,0)=\left(\frac{\d}{\d t}\Big|_{t=0}\E_{\gamma}\mathcal{S}_{t,W}f,\,\, \frac{\d}{\d t}\Big|_{t=0}\E_{\gamma}\mathcal{S}_{t,Z}g\right)
=0.
\end{equation}
And by definition of $L$, we have, for all $1\leq i\leq k$,
\begin{equation}\label{three2z}
\begin{aligned}
\frac{\partial L}{\partial\alpha_{i}}(0,0)&=\left(\frac{\d}{\d t}\Big|_{t=0}\E_{\gamma}\mathcal{S}_{t,W_{i}}f, \,\, 0 \right).\\
\frac{\partial L}{\partial\theta_{i}}(0,0)&=\left(  0, \,\, \frac{\d}{\d t}\Big|_{t=0}\E_{\gamma}\mathcal{S}_{t,W_{i}}g \right).\\
\end{aligned}
\end{equation}
From Lemma \ref{lemma30}, we conclude that the matrix of partial derivatives $DL$ of $L$ is a $(2k+1)\times 2k$ matrix of rank $2k$.  So, by the Implicit Function Theorem, there exists $\epsilon>0$ and a differentiable curve $\eta\colon(-\epsilon,\epsilon)\to\R^{2k+1}$ with $\eta(0)=0$, $\eta'(0)\neq0$ and $L(\eta(t))=0$ for all $t\in(-\epsilon,\epsilon)$.  The last property implies, by the Chain Rule, that
\begin{equation}\label{three3z}
0=\frac{\d}{\d t}\Big|_{t=0}L(\eta(t))
=\sum_{i=1}^{k}\frac{\partial L}{\partial\alpha_{i}}(0,0)\eta_{i}'(0)
+\sum_{i=1}^{k}\frac{\partial L}{\partial\theta_{i}}(0,0)\eta_{k+i}'(0)\qquad+\frac{\partial L}{\partial \beta}(0,0)\eta_{2k+1}'(0).\\
\end{equation}
The last term is zero by \eqref{three1z}.  Lemma \ref{lemma30} and \eqref{three2z} imply that set $\Big\{\frac{\partial L}{\partial\alpha_{i}}(0,0),\frac{\partial L}{\partial\theta_{i}}(0,0)\Big\}_{i=1}^{k}$ consists of $2k$ linearly independent vectors.  We conclude from \eqref{three3z} that $\eta_{i}'(0)=\eta_{i+k}'(0)=0$ for all $1\leq i\leq k$.  Since $\eta'(0)\neq0$, we conclude that $\eta_{2k+1}'(0)\neq0$.  (Observe $L(\eta(t))=0$ for all $t\in(-\epsilon,\epsilon)$, i.e. $\E_{\gamma}f=\E_{\gamma}g=0$ for all $t\in(-\epsilon,\epsilon)$, so the desired constraints hold for $f,g$ for all $t\in(-\epsilon,\epsilon)$.)

Taking another derivative of \eqref{three3z} and using $\eta_{i}'(0)=\eta_{i+k}'(0)=0$ for all $1\leq i\leq k$ along with $L(\eta(t))=0$ for all $t\in(-\epsilon,\epsilon)$ and \eqref{three1z},
\begin{equation}\label{three4}
0=\frac{\d^{2}}{\d t^{2}}\Big|_{t=0}L(\eta(t))
=\sum_{i=1}^{k}\frac{\partial L}{\partial\alpha_{i}}(0,0)\eta_{i}''(0)
+\sum_{i=1}^{k}\frac{\partial L}{\partial\theta_{i}}(0,0)\eta_{k+i}''(0)\qquad+\frac{\partial^{2} L}{\partial \beta^{2}}(0,0)(\eta_{2k+1}'(0))^{2}.\\
\end{equation}

Define $J_{0}(\alpha,\theta,\beta)\colonequals \underset{X\sim_{\rho}Y}{\E}\langle f_{\alpha,\beta}(X),\, g_{\theta,\beta}(Y)\rangle$ and $J_{1}(t)\colonequals J_{0}(\eta(t))$.  Then, since $f,g$ are optinally stable and $L(\eta(t))=0$ for all $t\in(-\epsilon,\epsilon)$,  we have $J_{1}'(0)=0$.  Since $\eta_{i}'(0)=\eta_{i+k}'(0)=0$ for all $1\leq i\leq k$, we have $0=J_{1}'(0)=\eta_{2k+1}'(0)\frac{\partial J_{0}}{\partial\beta}(0,0)$.  Since $\eta_{2k+1}'(0)\neq0$, we have $\frac{\partial J_{0}}{\partial\beta}(0,0)=0$.  Taking another derivative of $J_{1}$ and using optimality of $f,g$, we have
\begin{equation}\label{three5}
0\leq J_{1}''(0)=\sum_{i=1}^{k}\frac{\partial J_{0}}{\partial\alpha_{i}}(0,0)\eta_{i}''(0)
+\sum_{i=1}^{k}\frac{\partial J_{0}}{\partial\theta_{i}}(0,0)\eta_{k+i}''(0)\qquad+\frac{\partial^{2} J_{0}}{\partial \beta^{2}}(0,0)(\eta_{2k+1}'(0))^{2}.
\end{equation}
Here we again used $\eta_{i}'(0)=\eta_{i+k}'(0)=0$ for all $1\leq i\leq k$.  Lemma \ref{lemma34} says $\frac{\partial}{\partial\alpha_{i}}J_{0}=\langle\lambda, \frac{\partial}{\partial\alpha_{i}}L\rangle$ for all $1\leq i\leq k$ and $\frac{\partial}{\partial\theta_{i}}J_{0}=\langle\lambda, \frac{\partial}{\partial\theta_{i}}L\rangle$ for all $1\leq i\leq k$.  So, we can rewrite \eqref{three5} as
\begin{equation}\label{three6}
\begin{aligned}
0\leq J_{1}''(0)
&=\sum_{i=1}^{k}\eta_{i}''(0)\langle\lambda, \frac{\partial}{\partial\alpha_{i}}L\rangle
+\sum_{i=1}^{k}\eta_{k+i}''(0)\langle\lambda, \frac{\partial}{\partial\theta_{i}}L\rangle\qquad+\frac{\partial^{2} J_{0}}{\partial \beta^{2}}(0,0)(\eta_{2k+1}'(0))^{2}\\
&\stackrel{\eqref{three4}}{=}-(\eta_{2k+1}'(0))^{2}\Big\langle\lambda,\frac{\partial^{2} L}{\partial \beta^{2}}(0,0)\Big\rangle+\frac{\partial^{2} J_{0}}{\partial \beta^{2}}(0,0)(\eta_{2k+1}'(0))^{2}.
\end{aligned}
\end{equation}
Since $(\eta_{2k+1}'(0))^{2}>0$, we can factor it out of \eqref{three6} to get
$$
0\leq -\Big\langle\lambda,\frac{\partial^{2} L}{\partial \beta^{2}}(0,0)\Big\rangle+\frac{\partial^{2} J_{0}}{\partial \beta^{2}}(0,0).
$$
This inequality concludes the proof, since the right side is $Q(W,Z)$.
\end{proof}

\begin{lemma}[\embolden{Second Variation of Translations}, adapted from Lemma 6.17, {\cite{hwang21}}]\label{lemma36}
Fix $w\in\R^{n}$.  Let $W\colon\R^{n}\to\R^{n}$ be the constant vector field $W\colonequals w$.  Let $f,g$ be optimally stable.  Assume that
$$\frac{\d}{\d t}\Big|_{t=0}\E_{\gamma}\mathcal{S}_{t,W}f=\frac{\d}{\d t}\Big|_{t=0}\E_{\gamma}\mathcal{S}_{t,Z}g=0.$$
Then
$$Q(W,W)=2(1-1/\rho)\E_{\gamma}\vnormf{D_{w}T_{\sqrt{\rho}}f}^{2}.$$
\end{lemma}
\begin{proof}
Let $z\in\R^{n}$.  Let $Z\colon\R^{n}\to\R^{n}$ be the constant vector field $Z\colonequals z$.  From \eqref{sdef},  $\mathcal{S}_{t,W}f(x)=f(x-tw)$ for all $t\in\R$, $x\in\R^{n}$.  So,
\begin{equation}\label{three7}
\frac{\d^{2}}{\d t^{2}}\Big|_{t=0}\E_{\gamma}\mathcal{S}_{t,W}f=\E[D_{w}(D_{w}f)]
,\qquad\frac{\d^{2}}{\d t^{2}}\Big|_{t=0}\E_{\gamma}\mathcal{S}_{t,Z}g=\E[D_{z}(D_{z}g)].
\end{equation}
Using the product rule and Definition \ref{oudef},
\begin{flalign*}
&\frac{\d^{2}}{\d t^{2}}\Big|_{t=0}\underset{X\sim_{\rho}Y}{\E}\langle\mathcal{S}_{t,W}f(X),\, \mathcal{S}_{t,Z}g(Y)\rangle\\
&\qquad=\underset{X\sim_{\rho}Y}{\E}\langle D^{2}_{w,w}f(X), g(Y)\rangle
+\underset{X\sim_{\rho}Y}{\E}\langle f(X), D^{2}_{z,z}g(Y)\rangle
+2\underset{X\sim_{\rho}Y}{\E}\langle D_{w}f(X), D_{z}g(Y)\rangle\\
&\qquad=\E_{\gamma}\langle D^{2}_{w,w}f, T_{\rho}g\rangle
+\E_{\gamma}\langle T_{\rho}f, D^{2}_{z,z}g\rangle
+2\E_{\gamma}\langle D_{w}f, T_{\rho}D_{z}g\rangle\\
&\qquad=\E_{\gamma}\langle D^{2}_{w,w}f, T_{\rho}g\rangle
+\E_{\gamma}\langle T_{\rho}f, D^{2}_{z,z}g\rangle
+\frac{2}{\rho}\E_{\gamma}\langle D_{w}f, D_{z}T_{\rho}g\rangle.
\end{flalign*}

Substituting the definition of $Q$ \eqref{qdef}, using \eqref{three7}, then integrating by parts twice,
\begin{equation}\label{three8}
\begin{aligned}
Q(W,Z)
&=\E_{\gamma}\langle D^{2}_{w,w}f, T_{\rho}g - \lambda^{(2)}\rangle
+\E_{\gamma}\langle T_{\rho}f - \lambda^{(1)}, D^{2}_{z,z}g\rangle
+\frac{2}{\rho}\E_{\gamma}\langle D_{w}f, D_{z}T_{\rho}g\rangle\\
&=-\sum_{i=1}^{k}\E_{\gamma}D_{w}f_{i}\mathrm{div}_{\gamma}(w(T_{\rho}g_{i} - \lambda_{i}^{(2)}))
-\sum_{i=1}^{k}\E_{\gamma}D_{z}g_{i}\mathrm{div}_{\gamma}(z(T_{\rho}f_{i} - \lambda_{i}^{(1)}))\\
&\qquad\qquad\qquad\qquad\qquad\qquad\qquad\qquad
-\frac{2}{\rho}\sum_{i=1}^{k}\E_{\gamma}f_{i}\mathrm{div}_{\gamma}(wD_{z}T_{\rho}g_{i})\rangle\\
&=\sum_{i=1}^{k}\E_{\gamma}f_{i}\mathrm{div}_{\gamma}(w\mathrm{div}_{\gamma}(w(T_{\rho}g_{i} - \lambda_{i}^{(2)})))
+\sum_{i=1}^{k}\E_{\gamma}g_{i}\mathrm{div}_{\gamma}(z\mathrm{div}_{\gamma}(z(T_{\rho}f_{i} - \lambda_{i}^{(1)})))\\
&\qquad\qquad\qquad\qquad\qquad\qquad\qquad\qquad
-\frac{2}{\rho}\sum_{i=1}^{k}\E_{\gamma}f_{i}\mathrm{div}_{\gamma}(wD_{z}T_{\rho}g_{i})\rangle.
\end{aligned}
\end{equation}

We examine the penultimate terms.  From the First Variation Lemma \ref{lemma33}, a.s.
$$\vnormf{T_{\rho}f-\lambda^{(1)}}g=T_{\rho}f-\lambda^{(1)},\qquad\vnormf{T_{\rho}g-\lambda^{(2)}}f=T_{\rho}g-\lambda^{(2)}.$$
So, using first the product rule then Lemma \ref{lemma33},
\begin{flalign*}
\sum_{i=1}^{k}\E_{\gamma}f_{i}\mathrm{div}_{\gamma}(w\mathrm{div}_{\gamma}(w(T_{\rho}g_{i} - \lambda_{i}^{(2)})))
&=\sum_{i=1}^{k}\E_{\gamma}f_{i}\mathrm{div}_{\gamma}((T_{\rho}g_{i} - \lambda_{i}^{(2)})w\mathrm{div}_{\gamma}(w)\,\,+w D_{w}T_{\rho}g_{i}  )     \\
&=\sum_{i=1}^{k}\E_{\gamma}f_{i}\mathrm{div}_{\gamma}(\vnormf{T_{\rho}g_{i} - \lambda_{i}^{(2)}}f_{i}w\mathrm{div}_{\gamma}(w)\,\,+w D_{w}T_{\rho}g_{i}  ).  
\end{flalign*}
The penultimate term here is zero by \cite[Lemma 6.4]{hwang21}, and a similar calculation for the other term in \eqref{three8} allows us to rewrite \eqref{three8} as
\begin{flalign*}
Q(W,Z)
&=\sum_{i=1}^{k}\E_{\gamma}f_{i}\mathrm{div}_{\gamma}(w D_{w}T_{\rho}g_{i})
+\sum_{i=1}^{k}\E_{\gamma}g_{i}\mathrm{div}_{\gamma}(z D_{z}T_{\rho}f_{i})\\
&\qquad\qquad\qquad\qquad\qquad\qquad\qquad\qquad
-\frac{2}{\rho}\sum_{i=1}^{k}\E_{\gamma}f_{i}\mathrm{div}_{\gamma}(wD_{z}T_{\rho}g_{i})\rangle.
\end{flalign*}
Lemma \ref{lemma37}, rewrites this as
$$Q(W,Z)=-\E_{\gamma}\langle D_{w}T_{\sqrt{\rho}}f,\, D_{w}T_{\sqrt{\rho}}g\rangle
-\E_{\gamma}\langle D_{z}T_{\sqrt{\rho}}f,\, D_{z}T_{\sqrt{\rho}}g\rangle
+\frac{2}{\rho}\E_{\gamma}\langle D_{w}T_{\sqrt{\rho}}f,\, D_{z}T_{\sqrt{\rho}}g\rangle.$$

Using now Lemma \ref{lemma0}, we may assume that $g=-f$, so that
$$Q(W,Z)=\E_{\gamma}\langle D_{w}T_{\sqrt{\rho}}f,\, D_{w}T_{\sqrt{\rho}}f\rangle
+\E_{\gamma}\langle D_{z}T_{\sqrt{\rho}}f,\, D_{z}T_{\sqrt{\rho}}f\rangle
-\frac{2}{\rho}\E_{\gamma}\langle D_{w}T_{\sqrt{\rho}}f,\, D_{z}T_{\sqrt{\rho}}f\rangle.$$

Choosing $w=z$, we get
$$
Q(W,W)=2(1-1/\rho)\E_{\gamma}\vnormf{D_{w}T_{\sqrt{\rho}}f}^{2}.
$$
\end{proof}

\begin{lemma}[Lemma 6.6, {\cite{hwang21}}]\label{lemma37}
Let $w\in\R^{n}$.  Then
$$\sum_{i=1}^{k}\E_{\gamma}f_{i}\mathrm{div}_{\gamma}(w D_{z}T_{\rho}g_{i})
=-\E_{\gamma}\langle D_{w}T_{\sqrt{\rho}}g,\, D_{z}T_{\sqrt{\rho}}f\rangle.$$
\end{lemma}

\begin{proof}[Proof of Theorem \ref{thm9}]
Let $f,g\colon\R^{n}\to S^{k-1}$ be optimally stable.  Let $w\in\R^{n}$.  Define $W\colon\R^{n}\to\R^{n}$ so that $W(x)\colonequals w$ for all $x\in\R^{n}$.  Let $W$ satisfy
\begin{equation}\label{ten1}
\frac{\d}{\d t}\Big|_{t=0}\E_{\gamma}\mathcal{S}_{t,W}f=\frac{\d}{\d t}\Big|_{t=0}\E_{\gamma}\mathcal{S}_{t,W}g=0.
\end{equation}
Since $f=-g$ by Lemma \ref{lemma0}, \eqref{ten1} is equivalent to
\begin{equation}\label{ten2}
\frac{\d}{\d t}\Big|_{t=0}\E_{\gamma}\mathcal{S}_{t,W}f=0.
\end{equation}
Define $L\colon \R^{n}\to \R^{k}$ by $L(w)\colonequals\frac{\d}{\d t}|_{t=0}\E_{\gamma}\mathcal{S}_{t,W}f$.  Let $\mathrm{ker}(L)\colonequals\{w\in\R^{n}\colon L(w)=0\}$ denote the kernel of $L$.  Since $L$ is a linear map, the rank-nullity theorem implies that $\mathrm{ker}(L)$ has dimension at least $n-k$.

Let $w\in\mathrm{ker}(L)$.  Then Lemma \ref{lemma35} says $Q(W,W)\geq0$, while Lemma \ref{lemma36} says
$$Q(W,W)=2(1-1/\rho)\E_{\gamma}\vnormf{D_{w}T_{\sqrt{\rho}}f}^{2}\leq0.$$
We conclude that there exists a linear subspace $\mathrm{ker}(L)\subset\R^{n}$ of dimension at least $n-k$ with
$$\E_{\gamma}\vnormf{D_{w}T_{\sqrt{\rho}}f}^{2}=0,\qquad\forall\,w\in\mathrm{ker}(L).$$
That is, we may assume that $f$ is constant on $\mathrm{ker}(L)$, as noted e.g. in \cite[Lemma 6.6]{hwang21}.  That is, we may assume a priori that $f\colon\R^{k}\to\R^{k-1}$.  Theorem \ref{thm9} follows.
\end{proof}

\noindent
\textbf{Acknowledgement}.  We acknowledge and thank Yeongwoo Hwang, Joe Neeman, Ojas Parekh, Kevin Thompson and John Wright for helpful discussions.

\bibliographystyle{amsalpha}

\newcommand{\etalchar}[1]{$^{#1}$}
\def\polhk#1{\setbox0=\hbox{#1}{\ooalign{\hidewidth
  \lower1.5ex\hbox{`}\hidewidth\crcr\unhbox0}}} \def\cprime{$'$}
  \def\cprime{$'$}
\providecommand{\bysame}{\leavevmode\hbox to3em{\hrulefill}\thinspace}
\providecommand{\MR}{\relax\ifhmode\unskip\space\fi MR }
\providecommand{\MRhref}[2]{%
  \href{http://www.ams.org/mathscinet-getitem?mr=#1}{#2}
}
\providecommand{\href}[2]{#2}

\end{document}